%% file: sc-SCUC-arXiv.tex
\documentclass[journal,twoside,web]{ieeecolor}
\usepackage{generic}

\usepackage{xtab}
\usepackage{scrextend}
\usepackage{cite}

\usepackage{amsthm}

\usepackage{amsmath,amsfonts, amssymb}

\usepackage{bbm}
\usepackage{subfig}
\usepackage{graphicx}
\graphicspath{{../fig/}}

\usepackage{algorithm,algorithmic}% 

\newcommand{\FR}[1]{{\textcolor{black}{#1}}}

\theoremstyle{plain}
\newtheorem{thm}{Theorem}
\newtheorem{lem}{Lemma}

\newtheorem{cor}{Corollary}

\newtheorem{assumption}{Assumption}

\theoremstyle{definition}
\newtheorem{defn}{Definition}

\theoremstyle{remark}
\newtheorem{rem}{Remark}

\DeclareMathOperator{\Es}{\mathcal{E}}
\DeclareMathOperator{\Sc}{\mathcal{S}}
\DeclareMathOperator{\Ir}{\mathcal{R}}
\DeclareMathOperator{\N}{\mathcal{N}}
\DeclareMathOperator{\Iv}{\mathcal{I}}
\DeclareMathOperator{\SP}{\text{SP}}
\DeclareMathOperator{\opt}{\text{opt}}
\DeclareMathOperator{\opx}{\text{opx}}
\DeclareMathOperator*{\esssup}{ess\,sup}

\title{Computing Essential Sets for Convex and Non-convex Scenario Problems: Theory and Application}

\author{Xinbo Geng, Le Xie, and M. Sadegh Modarresi
\thanks{X. Geng and L. Xie are with the Department
of Electrical and Computer Engineering, Texas A\&M University, College Station, TX, 77840. Email:
\{xbgeng, le.xie\}@tamu.edu. M.S. Modarressi is with Burns \& McDonnell, Houston, TX, 77027. Email: msmodarresi@burnsmcd.com.
This work is supported in part by Power Systems Engineering Research Center (PSERC), and in part by NSF OAC-1934675, ECCS-1839616 and CCF-1934904.}% <-this % stops a space
}

\begin{document}
\maketitle

\begin{abstract}
\FR{The scenario approach is a general data-driven algorithm to chance-constrained optimization. It seeks the optimal solution that is feasible to a carefully chosen number of scenarios. A crucial step in the scenario approach is to compute the cardinality of essential sets, which is the smallest subset of scenarios that determine the optimal solution. This paper addresses the challenge of efficiently identifying essential sets. For convex problems, we demonstrate that the sparsest dual solution of the scenario problem could pinpoint the essential set. For non-convex problems, we show that two simple algorithms return the essential set when the scenario problem is non-degenerate. Finally, we illustrate the theoretical results and computational algorithms on security-constrained unit commitment (SCUC) in power systems. In particular, case studies  of chance-constrained SCUC are performed in the IEEE 118-bus system. Numerical results suggest that the scenario approach could be an attractive solution to practical power system applications.}
\end{abstract}


% \fi
\IEEEpeerreviewmaketitle

% \input{main}
% \vspace{-0.3cm}
\section{Introduction}
\label{sec:introduction}
\FR{
Decision making in the presence of uncertainties is an important problem in both theory and practice. \emph{Stochastic} optimization (SO) and \emph{robust} optimization (RO) are two common approaches for decision making under uncertainties. SO relies on probabilistic models to depict uncertainties and often optimizes the objective function in the presence of randomness \cite{birge_introduction_2011}. RO takes an alternative approach, in which the uncertainty model is typically set-based and deterministic \cite{bertsimas_theory_2011}.
This paper provides a perspective of decision making in uncertain environments through the lens of chance-constrained optimization (CCO), which is akin to both SO and RO \cite{geng_data-driven_2019-2}. The main distinction between CCO and SO/RO is the \emph{chance constraint} (see \eqref{form:cc_opt_function_cc} and \eqref{form:cc_opt_cc} in Section \ref{sec:introduction_to_the_scenario_approach}), which explicitly considers the feasibility of solutions under uncertainties.
}

\FR{
CCO has found successful applications in many different areas, e.g., control theory \cite{calafiore_scenario_2006}, chemical process \cite{henrion_stochastic_2001} and power systems \cite{modarresi_scenario-based_2018}.}
% and recently in machine learning \cite{xu_robust_2009,ben-tal_robust_2009,caramanis_14_2012,ben-tal_chance_2011,sra_optimization_2012,gabrel_recent_2014}.
% Both SO and RO have been successfully applied in various areas.
% For example, many problems in the theory and application of system and control , e.g., system identification [x], xxx [x], and xxx [x].  
% SO has found many successful applications in power system operations and planning problems.
% For instance, references \cite{takriti_stochastic_1996,wu_stochastic_2007,zheng_stochastic_2015} formulate and solve the \emph{stochastic unit commitment} problem, which minimizes the expected commitment and dispatch costs.  Recently, researchers in \cite{bertsimas_adaptive_2013} formulated and solved the \emph{robust unit commitment} problem, which minimizes the commitment and dispatch costs for the worst case in a predefined uncertainty set. 
\FR{
Since the first chance-constrained program was formulated in 1950s \cite{charnes_cost_1958},
% then was extensively studied in the following 50 years, e.g. \cite{charnes_chance-constrained_1959,charnes_deterministic_1963,kataoka_stochastic_1963,pinter_deterministic_1989,sen_relaxations_1992,prekopa_programming_1998,ruszczynski_stochastic_2003-1,ben-tal_robust_2009,prekopa_stochastic_1995}. Previously, 
many classical methods were proposed to deal with specific families of distributions, e.g., multivariate Gaussian distribution \cite{kataoka_stochastic_1963}, log-concave distributions \cite{prekopa_stochastic_1995}. Some breakthroughs and novel methods appeared in the past ten years, including 
(1) sample average approximation (SAA) \cite{sen_relaxations_1992,ruszczynski_probabilistic_2002,ahmed_solving_2008,luedtke_integer_2010,luedtke_sample_2008,tanner_iis_2010,ahmed_relaxations_2018-1,ahmed_nonanticipative_2017};
(2) safe approximation \cite{nemirovski_convex_2006,ben-tal_robust_2009,chen_cvar_2010,nemirovski_safe_2012,bertsimas_data-driven_2018,bertsimas_probabilistic_2019}; 
and (3) the scenario approach \cite{calafiore_uncertain_2005,campi_exact_2008,campi_scenario_2009,calafiore_random_2010,campi_wait-and-judge_2016,campi_general_2018}. Some key results and recent progresses are summarized below, a more comprehensive review is in \cite{geng_data-driven_2019-2}.}

\FR{SAA is closely related with stochastic programming. \cite{sen_relaxations_1992} first introduced the idea of approximating the violation probability using samples. Based on this idea, \cite{ruszczynski_probabilistic_2002,ahmed_solving_2008,luedtke_integer_2010} show that CCO problems can be approximated by  mixed integer programs (MIPs). Many recent results have been developed on the computational and theoretical aspects of SAA. \cite{luedtke_integer_2010} proposes different strong MIP formulations for SAA, and various techniques to accelerate solving SAA have been proposed, e.g, \cite{luedtke_integer_2010,tanner_iis_2010}. Reference \cite{ahmed_relaxations_2018-1} studies the case of finite distributions and discusses existing approximation approaches for constructing good feasible solution. \cite{ahmed_nonanticipative_2017} proposes new Lagrangian dual based formulations that are proved to be superior than standard MIP formulations.}

\FR{An alternative approach to CCO is safe approximation (or inner approximation), whose solutions are guaranteed to satisfy the chance constraint. The safe approximation is often convex problems that are tractable to solve, e.g., \cite{nemirovski_convex_2006,nemirovski_safe_2012,ben-tal_robust_2009,chen_cvar_2010}. This line of work is closely related with robust optimization. For example, the safe approximations in \cite{ben-tal_robust_2009,chen_cvar_2010,nemirovski_safe_2012,bertsimas_data-driven_2018} can be written as RO problems using uncertainty sets. Reference \cite{bertsimas_probabilistic_2019} further shows that RO problems with carefully constructed uncertainty sets could return optimal solutions that satisfy chance constraints with high confidence. \cite{bienstock_chance-constrained_2014} derives a conic formulation to approximate a CCO problem with linear constraints, and applies a cutting plane algorithm to solve it. Another related work is \cite{ahmed_convex_2014}, which develops outer approximations to CCO problems and obtain lower bounds on the optimal objective value. 
}

\FR{The scenario approach (also known as scenario approximation) is another popular approach to solve CCO problems. The scenario approach seeks the optimal solution that is feasible to a carefully chosen number of scenarios (sample complexity). Reference \cite{calafiore_uncertain_2005} first proved a lower bound on sample complexity. Later \cite{campi_exact_2008,campi_scenario_2009,calafiore_random_2010} significantly improved the sample complexity bound and show that the bound is tight for a special class of CCO problems. The classical scenario approach in \cite{calafiore_uncertain_2005,campi_exact_2008,campi_scenario_2009,calafiore_random_2010} require all constraints to be convex, thus not applicable to non-convex problems. A recent breakthrough\cite{campi_general_2018} extends the scenario approach theory towards non-convex problems. For both convex and non-convex scenario approach, there is a critical step in common: estimating or computing the cardinality of essential sets, which is defined as the smallest subset of scenarios that determine the optimal solution (see Definitions \ref{defn:support_scenario} and \ref{def:essential_set}).
}

\FR{This paper addresses the challenge of efficiently identifying the cardinality of essential sets for both convex and non-convex scenario problems. The main contributions of this paper are threefold.}

% (1) We contribute to the non-convex scenario approach theory by proving salient structural properties of non-convex scenario problems, which extends the classical results for convex scenario problems published in [23]. (2) We formulate c-SCUC, which is later reformulated to scenario-based SCUC (s-SCUC) and solved via the scenario approach. To the best of our knowledge, this paper is the first to solve c-SCUC using the scenario approach while considering critical constraints such as transmission limits. (3) We design efficient algorithms to explore the structural properties of s-SCUC, which enables rigorous guarantees on the optimal solution returned by the scenario approach.

\FR{From the theoretical perspective, we connect the convex scenario approach theory with recent progress under the non-convex setting, we reveal the connection and distinction among similar concepts defined in different contexts, i.e., support scenarios \cite{campi_exact_2008} and essential sets \cite{calafiore_random_2010} in the convex setting, irreducible and minimum support subsample \cite{campi_general_2018} in the non-convex setting. We also extend the classical results on the structure of convex scenario problems in \cite{calafiore_random_2010} towards non-convex problems.}

\FR{From the algorithmic perspective, we design algorithms that can efficiently identify essential sets. For convex problems, essential sets can be pinpointed using the sparsest dual solution (Theorem \ref{thm:support_scenario_multiplier}). For non-convex problems, we show that simple algorithms based on definitions return the essential set when the scenario problem is non-degenerate (Theorem \ref{thm:nondegenerate_unique_essential_irreducible}).}

\FR{From the application perspective, we provide a practical solution of chance-constrained Security-constrained Unit Commitment (c-SCUC) using the scenario approach. To the best of our knowledge, this paper is the first to apply the scenario approach on c-SCUC while providing theoretical guarantees on the feasibility of solution. Based on the special two-stage structure of scenario-based SCUC (s-SCUC), it is possible to design efficient algorithms to identify essential sets despite of the non-convexity of s-SCUC.}

The remainder of this paper is organized as follows. Section \ref{sec:introduction_to_the_scenario_approach} introduces the scenario approach for both convex and non-convex problems. Section \ref{sec:structural_properties_of_general_scenario_problems} studies the properties of essential sets in convex and non-convex scenario problems and proposes various algorithms to identify essential sets. Section \ref{sec:security_constrained_unit_commitment} formulates chance-constrained SCUC, which is solved via the scenario approach. Numerical results and discussions are in Section \ref{sec:case_study} and \ref{sec:discussions}, respectively. Section \ref{sec:concluding_remarks} presents the concluding remarks. All proofs are available at \cite{geng_computing_2019}.

The notations in this paper are standard.
All vectors are in the real field $\mathbf{R}$. We use $\mathbf{1}$ to represent an all-one vector of appropriate size. The transpose of a vector $a$ is $a^\intercal$.
The element-wise multiplication of the same-size vectors $a$ and $b$ is denoted by $a \circ b$. For instance, $[a_1;a_2] \circ [b_1;b_2] = [a_1 b_1; a_2 b_2]$.
Sets are in calligraphy fonts, e.g., $\mathcal{S}$. The cardinality of a set $\mathcal{S}$ is $|\mathcal{S}|$. Removal of element $i$ from set $\mathcal{N}$ is represented by $\mathcal{N}-i$. The essential supremum is $\esssup$.
% section introduction (end)
% \vspace{-0.2cm}
\section{Introduction to the Scenario Approach} % (fold)
\label{sec:introduction_to_the_scenario_approach}
% This section first provides an introduction to chance-constrained optimization. Section \ref{sub:the_scenario_approach_for_convex_problems} presents the main results of the scenario approach for convex problems. Recent progress in the scenario approach for non-convex problems are summarized in Section \ref{sub:the_scenario_approach_for_non_convex_problems}.

\subsection{Chance-constrained Optimization (CCO)} % (fold)
\label{sub:chance_constrained_optimization}
The typical formulation of CCO is in \eqref{form:cc_opt_function}.
\begin{subequations}
\label{form:cc_opt_function}
\begin{align}
  \min_x~ &  c^\intercal x \\
 % & \text{s.t.} && \overline{Q_{CB}^{(i)}} = \frac{1}{T} \sum_{t=1}^{T} Q_{CB}^{(i)}[t]  & &  \\
  \text{s.t.}~&  \mathbb{P}_{\xi} \Big( f(x,\xi) \le 0 \Big) \ge 1 - \epsilon \label{form:cc_opt_function_cc} \\
  & g(x) \le 0 \label{form:cc_opt_function_det}
\end{align}
\end{subequations}
We could write \eqref{form:cc_opt_function} in a more compact form by defining $\mathcal{X}_\xi:= \{x \in \mathbf{R}^n: f(x,\xi) \le 0\}$ and $\chi := \{x \in \mathbf{R}^n: g(x) \le 0 \}$
\begin{subequations}
\label{form:cc_opt}
\begin{align}
  \min_{x \in \chi}~ &  c^\intercal x  \\
 % & \text{s.t.} && \overline{Q_{CB}^{(i)}} = \frac{1}{T} \sum_{t=1}^{T} Q_{CB}^{(i)}[t]  & &  \\
  \text{s.t.}~&  \mathbb{P}_{\xi} \Big( x \in \mathcal{X}_\xi \Big) \ge 1 - \epsilon \label{form:cc_opt_cc}
\end{align}
\end{subequations}
Without loss of generality, we assume that the objective is a linear function of decision variables $x \in \mathbf{R}^d$ \cite{campi_scenario_2009}. Random vector $\xi \in \Xi$ denotes the source of uncertainties and $\Xi$ is the support of $\xi$. Deterministic constraints \eqref{form:cc_opt_function_det} are represented by set $\chi$ in \eqref{form:cc_opt}. Constraint \eqref{form:cc_opt_function_cc} or \eqref{form:cc_opt_cc} is the \emph{chance constraint}. The chance constraint \eqref{form:cc_opt_cc} requires the the inner constraint $x \in \mathcal{X}_\xi$ to be satisfied with probability at least $1 - \epsilon$, where the violation probability $\epsilon$ is typically a small number (e.g., $1\%, 5\%$). In \eqref{form:cc_opt_cc}, the set $\mathcal{X}_\xi$ depends on the realization of $\xi$ and the probability is taken with respect to $\xi$.

As discussed in Section \ref{sec:introduction}, many methods have been proposed to solve chance-constrained optimization problems. A detailed review and tutorial on chance-constrained optimization is in \cite{geng_data-driven_2019-2}. Compared with other methods, the scenario approach has many advantages such as computationally efficient and applicable for a broad range of optimization problems. 
% subsection chance_constrained_optimization (end)

\subsection{The Scenario Approach for Convex Problems} % (fold)
\label{sub:the_scenario_approach_for_convex_problems}
The scenario approach utilizes $N$ independent and identically distributed (i.i.d.) scenarios $\mathcal{N} := \{\xi^1,\xi^2,\cdots,\xi^N\}$ to \FR{approximate} the chance-constrained program \eqref{form:cc_opt_function} with the \emph{scenario problem} below:
\begin{subequations}
\label{form:scenario_problem_function}
\begin{align}
  \text{SP}(\mathcal{N}):~\min_{x} \quad &  c^\intercal x  \\
 % & \text{s.t.} && \overline{Q_{CB}^{(i)}} = \frac{1}{T} \sum_{t=1}^{T} Q_{CB}^{(i)}[t]  & &  \\
  \text{s.t. } & f(x,\xi^1) \le 0 \label{form:scenario_problem_scenario_function_1}  & :\mu^{1} \\ 
  & \qquad \vdots \nonumber \\
  & f(x,\xi^N) \le 0 \label{form:scenario_problem_scenario_function_N}  & :\mu^{N}  \\
  & g(x) \le 0 &: \lambda
  \end{align}
\end{subequations}
The scenario problem $\SP(\N)$ seeks the optimal solution $x_{\mathcal{N}}^*$ that is feasible for all $N$ scenarios. The Lagrangian multiplier associated with the $i$th scenario constraint $f(x,\xi^{i}) \le 0$ is denoted by $\mu^{i} \in \mathbf{R}^m$. We can write the scenario problem $\SP(\N)$ in a similar way with \eqref{form:cc_opt} by defining $\mathcal{X}_{i}:= \{x \in \mathbf{R}^n: f(x, \xi^i) \le 0\}$.
\begin{subequations}
\label{form:scenario_problem_set}
\begin{align}
  \text{SP}(\mathcal{N}):~\min_{x \in \chi}~&  c^\intercal x  \\
 % & \text{s.t.} && \overline{Q_{CB}^{(i)}} = \frac{1}{T} \sum_{t=1}^{T} Q_{CB}^{(i)}[t]  & &  \\
  \text{s.t.}~&  x \in \cap_{i=1}^N \mathcal{X}_{i} \label{form:scenario_problem_scenario}
\end{align}
\end{subequations}
 \begin{defn}[Violation Probability]
The \emph{violation probability} of a candidate solution $x^\diamond$ is defined as the probability that $x^\diamond$ is infeasible:
\begin{equation}
\mathbb{V}(x^\diamond) := \mathbb{P}_\xi\big( x^\diamond \notin \mathcal{X}_\xi \big).
\end{equation}
\end{defn}
The scenario approach theory aims at answering the following \emph{sample complexity} question: what is the smallest sample size $N$ such that $x_{\mathcal{N}}^*$ is feasible (i.e., $\mathbb{V}(x_\mathcal{N}^*) \le \epsilon$) to the original chance-constrained program \eqref{form:cc_opt}? Reference \cite{campi_exact_2008,calafiore_random_2010} provide in-depth analysis based on the concept of support scenarios.
\begin{defn}[Support Scenario \cite{campi_exact_2008,calafiore_random_2010}]
\label{defn:support_scenario}
Scenario $\xi^i$ is a \emph{support scenario} for the scenario problem $\SP(\N)$ if its removal changes the solution of $\SP(\N)$.
\end{defn}
Let $x_{\N}^*$ and $x_{\N-i}^*$ stand for the optimal solution to scenario problems $\SP(\N)$ and $\SP(\N-i)$, respectively. Then scenario $\xi^i$ is a support scenario if $c^\intercal x_{\N-i}^* < c^\intercal x_{\N}^*$. We use $\Sc(\N)$ ($\Sc$ in short) to represent the set of all support scenarios of $\SP(\N)$.
\begin{defn}[Non-degenerate Scenario Problem \cite{campi_exact_2008,calafiore_random_2010}]
\label{defn:degenerate-convex}
Let $x_{\mathcal{N}}^*$ and $x_{\mathcal{S}}^*$ be the optimal solutions to the scenario problems $\SP(\N)$ and $\SP(\Sc)$, respectively. The scenario problem $\SP(\N)$ is said to be \emph{non-degenerate}, if $c^\intercal x_{\mathcal{N}}^* = c^\intercal x_{\mathcal{S}}^*$.
\end{defn}
\begin{assumption}[Non-degeneracy\cite{campi_exact_2008,calafiore_random_2010}]
\label{ass:non-degeneracy}
For every $N$, the scenario problem $\SP(\N)$ is non-degenerate with probability $1$ with respect to scenarios $\N = \{\xi^{1},\xi^{2},\cdots,\xi^{N}\}$.
\end{assumption}

\begin{assumption}[Feasibility \cite{campi_exact_2008}]
\label{ass:feasibility_uniqueness}
Every scenario problem $\SP(\N)$ is feasible, and its feasibility region has a non-empty interior. The optimal solution $x_{\mathcal{N}}^\ast$ of $\SP(\N)$ exists.
\end{assumption}
% The feasibility assumption can be further relaxed, more details can be found in \cite{calafiore_random_2010}.
\begin{defn}[Helly's Dimension \cite{calafiore_random_2010}]
Helly's dimension of the scenario problem $\SP(\N)$ is the smallest integer $h$ that 
$h \ge \esssup_{\N \subseteq \Xi^N} |\Sc(\N)|$ holds for any finite $N \ge 1$, where $|\Sc(\N)|$ is the number of support scenarios.
\end{defn}
Theorem \ref{thm:exact_feasibility_scenario_approach} presents one of the most important results in the scenario approach theory, which is based on the non-degeneracy and feasibility assumptions.
\begin{thm}[Exact Feasibility \cite{campi_exact_2008,calafiore_random_2010}]
\label{thm:exact_feasibility_scenario_approach}
Under Assumptions \ref{ass:non-degeneracy} (non-degeneracy) and \ref{ass:feasibility_uniqueness} (feasibility), let $x_{\N}^\ast$ be the optimal solution to the scenario problem $\SP(\N)$, it holds that
\begin{equation}
\label{eqn:exact_distribution_violation_prob}
  \mathbb{P}^N \Big( \mathbb{V}(x_{\N}^\ast)  > \epsilon \Big) \le \sum_{i=1}^{h-1} \binom{N}{i} \epsilon^i (1- \epsilon )^{N-i}.
\end{equation}
The probability $\mathbb{P}^N$ is taken with respect to $N$ random scenarios $\mathcal{N} = \{\xi^i\}_{i=1}^N$, and $h$ is the Helly's dimension of $\SP(\N)$.
\end{thm}
% The non-degeneracy assumption \ref{ass:non-degeneracy} lies at the heart of the scenario approach theory,
Stronger results without the feasibility assumption are in \cite{calafiore_random_2010,campi_exact_2008}.
Based on Theorem \ref{thm:exact_feasibility_scenario_approach}, the scenario approach answers the sample complexity question in Corollary \ref{cor:prior_sample_complexity_fully_supported}.
\begin{cor}[Sample Complexity \cite{campi_exact_2008,calafiore_random_2010}]
\label{cor:prior_sample_complexity_fully_supported}
Under Assumptions \ref{ass:non-degeneracy} (non-degeneracy) and \ref{ass:feasibility_uniqueness} (feasibility), given a violation probability $\epsilon \in (0,1)$ and a confidence parameter $\beta \in (0,1)$, if we choose the smallest number of scenarios $N$ such that 
\begin{equation}
  \sum_{i=0}^{h-1} \binom{N}{i} \epsilon^i (1- \epsilon)^{N-i} \le \beta,
\end{equation}
then it holds that
\begin{equation}
  \mathbb{P}^N\Big ( \mathbb{V}(x_{\mathcal{N}}^\ast) \le \epsilon \Big) \ge 1 - \beta,
\end{equation}  
where $x_{\mathcal{N}}^\ast$ is the optimal solution to $\SP(\N)$, and $h$ is the Helly's dimension of $\SP(\N)$ ($0 \le h \le N$). 
\end{cor}
The scenario approach is essentially a randomized algorithm to find a feasible solution to chance-constrained optimization problems. The randomness of the scenario approach comes from drawing i.i.d. scenarios. The confidence parameter $\beta$ quantifies the risk of failure due to drawing scenarios from a ``bad'' set. Corollary \ref{cor:prior_sample_complexity_fully_supported} shows that by choosing a proper number of scenarios, the corresponding optimal solution $x_{\mathcal{N}}^*$ is feasible (i.e., $\mathbb{V}(x_{\N}^*) \le \epsilon$) with confidence at least $1 - \beta$.
\begin{assumption}[Convexity]
\label{ass:convexity}
The deterministic constraint $g(x) \le 0$ is convex, and the random constraint $f(x,\xi)$ is convex in $x$ for every instance of $\xi$. In other words, the sets $\chi$ and $\mathcal{X}_i$s in \eqref{form:scenario_problem_set} are convex.
\end{assumption}
\begin{thm}[\cite{calafiore_uncertain_2005,calafiore_random_2010}]
\label{thm:helly_dimension_convex}
Under Assumption \ref{ass:feasibility_uniqueness} and \ref{ass:convexity}, 
the number of support scenarios $|\Sc|$ for $\SP(\N)$ is at most $n$. 
In other words, $h \le n$, where $n$ is the number of decision variables $x \in \mathbf{R}^n$ and $h$ is Helly's dimension. 
\end{thm}
For convex scenario problems $\SP(\N)$, we could replace $h$ by $n$ in Theorem \ref{thm:exact_feasibility_scenario_approach} and Corollary \ref{cor:prior_sample_complexity_fully_supported}. This leads to the classical results of the scenario approach in \cite{calafiore_uncertain_2005,campi_exact_2008,campi_scenario_2009,calafiore_random_2010}.
\begin{rem}[Towards Non-convexity]
Theorem \ref{thm:exact_feasibility_scenario_approach} and Corollary \ref{cor:prior_sample_complexity_fully_supported} do not assume convexity of $f(x,\xi)$ and $g(x)$. In theory, Theorem \ref{thm:exact_feasibility_scenario_approach} and Corollary \ref{cor:prior_sample_complexity_fully_supported} are applicable for non-convex scenario problems if a feasible non-convex $\SP(\N)$ is proved to be non-degenerate with probability $1$ (e.g., \cite{geng_chance-constrained_2019}). In practice, however, the scenario approach was considered \emph{not applicable} for non-convex problems. Comprehensive analysis are presented in Section \ref{sub:the_scenario_approach_for_non_convex_problems}.
\end{rem}
% subsection the_scenario_approach_for_convex_problems (end)

\subsection{The Scenario Approach for Non-convex Problems} % (fold)
\label{sub:the_scenario_approach_for_non_convex_problems}
The scenario approach was considered \emph{not applicable} for non-convex problems for the following three reasons: (1) non-convexity causes degeneracy; (2) non-trivial bounds on $|\Sc|$ may not exist for non-convex
$\SP(\N)$; and (3) it is computationally intractable to find optimal solutions.

First, degeneracy is a common issue for non-convex problems, e.g., the scenario-based SCUC problem in Section \ref{sub:s_scuc_is_degenerate}. 
Since the non-degeneracy assumption \ref{ass:non-degeneracy} lies at the heart of the scenario approach theory, almost all results in the literature are for non-degenerate problems.

Second, it is almost impossible to prove non-trivial and practical bounds on the number of support scenarios $|\Sc|$ for non-convex problems.
Reference \cite{campi_general_2018} presents one extreme case, in which every scenario is a support scenario thus $|\Sc|=N$ \footnote{Using the trivial bound $|\Sc| \le N$, Theorems \ref{thm:exact_feasibility_scenario_approach} and \ref{thm:scenario_theory_nonconvex_posterior} provide guarantees $\mathbb{P}( \mathbb{V}(x^*_{\N}) > \epsilon)\le 1$, which is useless.}. In addition, a loose bound typically leads to an astronomical sample complexity $N$, which make the scenario approach unpractical. For instance, loose bounds on $|\Sc|$ for scenario-based unit commitment will require $10^3 \sim 10^4$ times more scenarios than necessary \cite{geng_chance-constrained_2019}.

Furthermore, the most attractive feature of convex optimization is that any \emph{local} minimum is a \emph{global} minimum. And there exist a broad family of efficient algorithms that compute global optimal solutions for convex problems. Hence, $x_{\N}^*$ in Section \ref{sub:the_scenario_approach_for_convex_problems} refers to the \emph{global} optimal solution by default.
% Additionally, Theorem \ref{thm:helly_dimension_convex} stands on the fact that $x_{\N}^*$ is globally optimal.
It is worth noting that $x_{\N}^*$ is solely determined by the scenario problem $\SP(\N)$ and it is \emph{not} algorithm-dependent.

For non-convex problems, however, it is often computationally intractable to find global optimal solutions. There are many algorithms that are capable of finding \emph{local} optimal solutions in a relatively short time. Therefore, it is more reasonable and practical to analyze the characteristics of \emph{local} solutions for non-convex scenario problems. 
Algorithm  $\mathbb{A}: \Xi^N \rightarrow \mathbf{R}^n$ stands for the process of finding solutions to $\SP(\N)$, e.g., primal-dual interior-point method.
We use $\opx_\mathbb{A}(\N)$ to represent a (possibly suboptimal) solution to $\SP(\N)$ obtained via algorithm $\mathbb{A}$. The corresponding optimal objective value is denoted by $\opt_\mathbb{A}(\N)$. The subscript $\mathbb{A}$ emphasizes the fact that the solution is algorithm-dependent. And we use $\SP_\mathbb{A}(\N)$ to represent a scenario problem solved by algorithm $\mathbb{A}$.

Consequently, the scenario approach was considered \emph{not applicable} for non-convex problems until very recently. By removing the non-degeneracy assumption and analyzing any feasible solutions of non-convex scenario problems, reference \cite{campi_general_2018} develops a general theory for the scenario approach. This subsection summarizes its key results.

Identical to the convex case in Section \ref{sub:the_scenario_approach_for_convex_problems}, the scenario approach \FR{approximates} \eqref{form:cc_opt} to the scenario problem \eqref{form:scenario_problem_set} using $N$ scenarios $\mathcal{N} = \{\xi^1,\xi^2,\cdots, \xi^N\}$ for non-convex problems. The sets $\chi$ and $\mathcal{X}_\xi$ here could be non-convex.
% This procedure is depicted by algorithm $\mathbb{A}_N: \Xi^N \rightarrow \mathbf{R}^n$, which maps a set of scenarios $\N$ to a solution $\opx_\mathbb{A}(\N)$.
\begin{defn}[Invariant Set]
\label{def:invariant_set}
Let $\opt_\mathbb{A}(\mathcal{M})$ be the optimal value of $\SP(\mathcal{M})$ found by algorithm $\mathbb{A}$ for a scenario problem $\SP(\mathcal{M})$.
A set of scenarios $\mathcal{I}$ is an invariant (scenario) set for $\SP_\mathbb{A}(\N)$ if $\opt_\mathbb{A}(\mathcal{I}) = \opt_\mathbb{A}(\N)$.
\end{defn}
The concept of invariant set is an extension of support scenarios for (possibly degenerate) non-convex scenario problems. A trivial invariant set is $\Iv = \N$. Algorithm $\mathbb{B}: \Xi^N \rightarrow \Iv$ represents the process of finding non-trivial invariant sets. Examples of Algorithm $\mathbb{B}$ can be found in Section \ref{sec:structural_properties_of_general_scenario_problems} and Appendix \ref{sec:algorithms}.

\begin{thm}[Posterior Guarantees for Non-convex Scenario Problems \cite{campi_general_2018}\footnote{Theorem \ref{thm:scenario_theory_nonconvex_posterior} is a simplified version of the main result in \cite{campi_general_2018}, the feasibility assumption \ref{ass:feasibility_uniqueness} is a simplified version of the admissible assumption in \cite{campi_general_2018}.}]
\label{thm:scenario_theory_nonconvex_posterior}
Suppose Assumption \ref{ass:feasibility_uniqueness} (feasibility) holds true and $\beta \in (0,1)$ is given. Algorithm $\mathbb{A}$ solves the scenario problem $\SP(\N)$ and obtains an optimal solution $\opx_\mathbb{A}(\N)$. Algorithm $\mathbb{B}$ finds an invariant set $\mathcal{I}$ of cardinality $|\Iv|$. The following probabilistic guarantee holds
\begin{equation*}
  \mathbb{P}^N \Big( \mathbb{V}\big(\opx_\mathbb{A}(\N)\big) \le \epsilon(N,|\Iv|,\beta) \Big) \ge 1-\beta,
\end{equation*}
where the function $\epsilon(k,N,\beta)$ is defined as
\begin{equation}
\label{eqn:specified_epsilon_function}
\epsilon(N,k,\beta) := 
\begin{cases}
  1 & \text{if } k = N,\\
  1 - \Big( \frac{\beta}{ N \binom{N}{k}} \Big)^{\frac{1}{N-k}} & \text{otherwise.}
\end{cases}
\end{equation}
\end{thm}
\begin{lem}
\label{lem:monotonicity_specified_epsilon_function}
The $\epsilon(N,k,\beta)$ function defined in \eqref{eqn:specified_epsilon_function} has the following properties:
(1) $\epsilon(N,k,\beta)$ is monotonically decreasing in $\beta$;
(2) $\epsilon(N,k,\beta)$ is monotonically increasing in $k$;
(3) $\epsilon(N,k,\beta)$ is monotonically decreasing in $N$.
\end{lem}
In order to achieve an $\epsilon$-level solution with confidence $1-\beta$, Lemma \ref{lem:monotonicity_specified_epsilon_function} shows that the least conservative result (i.e., smallest sample complexity $N$) is achieved with the invariant set of minimal cardinality, which is defined as an essential set.
\begin{defn}[Essential Set \cite{calafiore_random_2010}]
\label{def:essential_set}
A set of scenarios $\Es \subseteq \N$ is an \emph{essential} (scenario) set for $\SP_\mathbb{A}(\N)$ if 
\begin{equation}
  \Es := \arg \min \{|\Es|: \opt_\mathbb{A}(\Es) = \opt_\mathbb{A}(\N), \Es \subseteq \N \}.
\end{equation}
In other words, $\Es$ is an invariant set of minimal cardinality.
\end{defn}
One key step in the non-convex scenario approach is designing algorithms $\mathbb{B}$ to search for essential sets. Section \ref{sec:structural_properties_of_general_scenario_problems} reveals the structure of general non-convex scenario problems, which lays the cornerstone for algorithms to obtain essential sets. Section \ref{sec:structural_properties_of_general_scenario_problems} also gives one example of designing more efficient algorithms by exploiting the structural properties of specific problems.
% \XG{Give an example to show that essential set, support set, and irreducible set could be different.}
% subsection the_scenario_approach_for_non_convex_problems (end)
% section introduction_to_the_scenario_approach (end)

\section{Computing Essential Sets for Convex and Non-convex Scenario Problems} % (fold)
\label{sec:structural_properties_of_general_scenario_problems}
Searching for essential sets is an important step in the non-convex scenario approach. However, the only known general algorithm to obtain essential sets is enumerating all $2^N$ possibilities by solving $2^N$ non-convex problems. This implies that searching for essential sets is in general computationally prohibitive. Section \ref{sub:non_convex_scenario_problems} first demonstrates the structural properties for general non-convex scenario problems, and provide conditions in which finding essential sets is relatively easier. Section \ref{sub:convex_scenario_problems} reveals the connection between non-convex and convex scenario problems. Section \ref{sub:two_stage_scenario_problems} studies a structured two-stage scenario problem, and illustrates an efficient algorithm to track down essential sets for two-stage scenario problems.
\subsection{Non-convex Scenario Problems} % (fold)
\label{sub:non_convex_scenario_problems}
Instead of solving $2^N$ non-convex problems to obtain essential sets, there are two ideas to track down invariant sets with small cardinalities (not necessarily essential): (1) removing each scenario and checking if the objective changes, this idea leads to the definition of support sets; (2) removing scenarios one by one, until the scenario set cannot be further reduced, this leads to the definition of irreducible set.
\begin{defn}[Support Scenario of $\SP_{\mathbb{A}}(\N)$]
\label{defn:support_scenario_non_convex}
Scenario $\xi^i \in \N$ is a \emph{support scenario} for the scenario problem $\SP_\mathbb{A}(\N)$ if its removal changes the solution $\opt_\mathbb{A}(\N)$ of $\SP_\mathbb{A}(\N)$. The set of support scenarios (support set in short) is denoted by $\Sc_\mathbb{A}$. 
\end{defn}
\begin{defn}[Irreducible Set]
\label{def:irreducible_set}
A scenario set $\Ir \subseteq \N$ for $\SP_\mathbb{A}(\N)$ is \emph{irreducible}, if (1) it is invariant, i.e., $\opt_\mathbb{A}(\Ir) = \opt_\mathbb{A}(\N)$; and (2) $\opt_\mathbb{A}(\Ir-s) < \opt_\mathbb{A}(\Ir) = \opt_\mathbb{A}(\N)$ for any $s \in \Ir$.
\end{defn}
\begin{assumption}[Monotonicity]
\label{ass:monotone_algorithm}
Let $\mathbb{A}: \Xi^N \rightarrow \mathbf{R}^n$ be an algorithm to obtain an optimal solution of a scenario problem $\SP(\N)$, whose optimal objective value is represented by $\opt_\mathbb{A}(\N)$. We assume that the algorithm $\mathbb{A}$ always satisfies $\opt_\mathbb{A}(\mathcal{M}) \le \opt_\mathbb{A}(\N)$ if $\mathcal{M} \subseteq \N$.
\end{assumption}
Assumption \ref{ass:monotone_algorithm} is indeed a weak assumption. Considering two scenario problems $\SP(\N)$ and $\SP(\mathcal{M})$ with $\mathcal{M} \subseteq \N$. Because the optimal solution to $\SP(\N)$ will be always feasible to $\SP(\mathcal{M})$, algorithm $\mathbb{A}$ could use $\opx_\mathbb{A}(\N)$ as a starting point and obtain solution $\opx_\mathbb{A}(\mathcal{M})$ that is not worse than $\opx_\mathbb{A}(\N)$.
\begin{lem}[Modified Lemma 2.10 of \cite{calafiore_random_2010}]
\label{lem:support_is_subset_invariant_nonconvex}
Suppose algorithm $\mathbb{A}$ satisfies Assumption \ref{ass:monotone_algorithm}.
Let $\Iv$ be any \emph{invariant} set for a (possibly non-convex) scenario problem $\SP_\mathbb{A}(N)$ and $\Sc$ stands for its support set, then $\Sc \subseteq \Iv$. Since any essential set $\Es$ or irreducible set $\Ir$ is also invariant, then $\Sc \subseteq \Es$ and $\Sc \subseteq \Ir$.
\end{lem}
Lemma \ref{lem:support_is_subset_invariant_nonconvex} reveals the key relationship among the support set, essential and irreducible sets, and it lays the foundation of more important observations in Theorem \ref{thm:nondegenerate_unique_essential_irreducible}.
Lemma \ref{lem:support_is_subset_invariant_nonconvex} is a generalized version of Lemma 2.10 in \cite{calafiore_random_2010}, which proved similar results for convex scenario problems. The importance of Lemma \ref{lem:support_is_subset_invariant_nonconvex} is to show that the key assumption for such structural properties is the monotonicity of algorithm $\mathbb{A}$, instead of convexity (Assumption \ref{ass:convexity} in \cite{calafiore_random_2010}).

\FR{Figure \ref{fig:case3-degenerate-example} in Section \ref{sub:s_scuc_is_degenerate} provides an illustrative example of these definitions. The support set is $\{1\}$ since only removing scenario $1$ changes the optimal solution. Both $\{1,2\}$ and $\{1,3\}$ are irreducible sets and essential sets. This example shows that essential sets and irreducible sets could be non-unique. There are three invariant sets $\{1,2,3\}$, $\{1,2\}$ and $\{1,3\}$, all of them contain the support set $\{1\}$. This verifies Lemma \ref{lem:support_is_subset_invariant_nonconvex}.} 
For general (non-convex) scenario problems, the support set $\Sc$, essential set $\Es$ and irreducible set $\Ir$ are different. Under certain circumstances, these three concepts are interchangeable. Such circumstances are depicted by an extended definition of non-degeneracy for non-convex scenario problems.
\begin{defn}[Non-degeneracy of $\SP_\mathbb{A}(\N)$]
\label{defn:degenerate-nonconvex}
For a general scenario problem $\SP_\mathbb{A}(\N)$, let $\N$ stand for the set of all $N$ scenarios and $\Sc$ denote the support (scenario) set. The scenario problem $\SP_{\mathbb{A}}(\N)$ is said to be \emph{non-degenerate}, if $\opt_\mathbb{A}(\N) = \opt_\mathbb{A}(\Sc)$.
\end{defn}
\begin{cor}
\label{cor:nondegenerate_uniqueness}
Consider a (possibly non-convex) scenario problem $\SP_\mathbb{A}(\N)$ and an algorithm $\mathbb{A}$ satisfying Assumption \ref{ass:monotone_algorithm}.  If $\SP_\mathbb{A}(\N)$ is non-degenerate, then (1) it has a unique essential set $\Es = \Sc$; and (2) it has a unique irreducible set $\mathcal{R} = \Sc$.
\end{cor}
% \XG{This results overlap with Corollary below, and are not very useful.}
\begin{lem}
\label{lem:removal_non_support}
\FR{Consider a (possibly non-convex) scenario problem $\SP_\mathbb{A}(\N)$ and an algorithm $\mathbb{A}$ satisfying Assumption \ref{ass:monotone_algorithm}. Suppose $k$ is not a support scenario for $\SP_\mathbb{A}(\N)$, then  }
\begin{equation}
  \Sc(\N) \subseteq \Sc(\N - k)
\end{equation}
\end{lem}
\FR{With Lemma \ref{lem:removal_non_support}, we can proved Theorem \ref{thm:nondegenerate_unique_essential_irreducible}, which are stronger than Corollary \ref{cor:nondegenerate_uniqueness}.}
\begin{thm}
\label{thm:nondegenerate_unique_essential_irreducible}
Consider a (possibly non-convex) scenario problem $\SP_\mathbb{A}(\N)$ and an algorithm $\mathbb{A}$ satisfying Assumption \ref{ass:monotone_algorithm}. The following three statements are equivalent:
(1) $\SP_{\mathbb{A}}(\N)$ is non-degenerate;
(2) $\SP_{\mathbb{A}}(\N)$ has a unique irreducible set $\Ir$;
and (3) $\SP_{\mathbb{A}}(\N)$ has a unique essential set $\Es$.
\end{thm}
Theorem \ref{thm:nondegenerate_unique_essential_irreducible} provides key insights in designing efficient algorithms $\mathbb{B}$. For non-convex problems, even if Assumption \ref{ass:non-degeneracy} does not always hold, $\SP_{\mathbb{A}}(\N)$ might still be non-degenerate in many instances (e.g., s-SCUC is non-degenerate in 192 out of 200 instances in Section \ref{sub:finding_essential_sets_for_s_scuc}). For those non-degenerate scenario problems, Theorem \ref{thm:nondegenerate_unique_essential_irreducible} shows that we are able to find the essential set by solving only $N$ instead of $2^N$ non-convex problems. Section \ref{sub:two_stage_scenario_problems} shows that the computational burden to obtain essential sets can be further reduced by exploiting the structure of specific problems.
\begin{rem}[Greedy Algorithms for Non-degenerate Problems]
\label{rem:find_essential_sets_nondegenerate}
When a scenario problem is non-degenerate, we can obtain the (unique) essential set by searching for the support set or irreducible set (Corollary \ref{cor:nondegenerate_uniqueness}). 
Algorithms of finding an irreducible set (Algorithm \ref{alg:find-irreducible-set} in Appendix \ref{sec:algorithms}) or the support set (e.g., Algorithm \ref{alg:find-support-set} in Appendix \ref{sec:algorithms}) are based on definitions. \FR{If the non-convex scenario problem is structured, e.g., \eqref{opt:two-stage-scenario-problem}, then more efficient algorithms are possible.}
% More discussions on finding the support set are in Remark \ref{rem:find_support_scenarios_convex}. \XG{These two are greedy algorithms.}
 % i.e., removing scenarios one by one until no single scenario can be further reduced.
% Remark \ref{rem:find_support_scenarios_convex} discusses the pursuit of support scenarios for  both convex and non-convex scenario problems.
\end{rem}
It is worth mentioning that similar results in this section for convex problems were first proved in \cite{calafiore_random_2010}. Section \ref{sub:non_convex_scenario_problems} can be regarded as an extension of classical results in \cite{calafiore_random_2010} towards non-convex scenario problems.

% subsection non_convex_scenario_problems (end)
\subsection{Convex Scenario Problems} % (fold)
\label{sub:convex_scenario_problems}
For convex scenario problems $\SP(\N)$, any local minimum is a global minimum. And there are a broad range of algorithms to look for global optimal solutions. In the convex setting, we assume any algorithm $\mathbb{A}$ returns global optimal solutions to $\SP(\N)$ by default. In Section \ref{sub:the_scenario_approach_for_convex_problems} and \ref{sub:convex_scenario_problems}, we replace $\opx_\mathbb{A}(\N)$ and $\opt_\mathbb{A}(\N)$ by $x_\mathcal{N}^*$ and $c^\intercal x_\mathcal{N}^*$, respectively. We also remove subscripts $\mathbb{A}$ since the definitions of support set, invariant set and essential set for convex problems are no longer algorithm-dependent.
Furthermore, since the support set, invariant set and essential set are the same for non-degenerate convex scenario problems (Theorem \ref{thm:nondegenerate_unique_essential_irreducible}), we use the term \emph{essential set} for consistency. 
\begin{lem}[Monotonicity]
\label{lem:monotonicity}
Let $x_{\N}^*$ and $x_\mathcal{\mathcal{M}}^*$ stand for the optimal solution to the convex scenario problems $\SP(\N)$ and $\SP(\mathcal{M})$, respectively. Then $c^\intercal x_{\mathcal{M}}^* \le c^\intercal x_{\N}^*$ if $\mathcal{M} \subseteq \mathcal{N}$. 
\end{lem}
Because $x_{\N}^*$ is always feasible to $\SP(\mathcal{M})$ and $x_{\mathcal{M}}^*$ is globally optimal, it is obvious that $c^\intercal x_{\mathcal{M}}^* \le c^\intercal x_{\N}^*$.
Lemma \ref{lem:monotonicity} shows that any algorithm obtaining global optimal solutions will automatically satisfy Assumption \ref{ass:monotone_algorithm}. Therefore, all results in Section \ref{sub:non_convex_scenario_problems} hold for convex scenario problems.

\FR{One attractive feature of convex scenario problems is that the number of support scenarios is bounded by the number of decision variables $n$ (Theorem \ref{thm:helly_dimension_convex}). Thus many papers simply replace Helly's dimension $h$ with $n$ when computing sample complexity $N$. However, this simple approach often causes extremely conservative results. For example, \cite{modarresi_scenario-based_2018} reported the number of support scenarios in look-ahead economic dispatch is $3\sim5$, which is much less than the number of decision variables ($864$). It is necessary to compute the support scenarios (i.e., the essential set as shown in Corollary \ref{cor:nondegenerate_uniqueness}) and improve the theoretical guarantees.}

\FR{Unlike the non-convex setting in Section \ref{sub:non_convex_scenario_problems}, which requires solving the scenario problem $N$ times even in the easy cases. For convex scenario problems, we show that the essential set can be pinpointed with dual variables, which is much more computationally efficient.}

\FR{The first idea is to check the slackness of constraints $f(x,\xi^i) \le 0$. Lemma \ref{lem:inactive_are_not_support_scenario} shows that inactive scenarios (i.e., $f(x,\xi^i) < 0$) cannot be support scenarios. This simple idea could significantly reduce the search space in practice.} 
\begin{lem}
\label{lem:inactive_are_not_support_scenario}
\FR{
Consider a \emph{non-degenerate} scenario problem $\SP(\N)$ under Assumptions \ref{ass:feasibility_uniqueness} (feasibility) and \ref{ass:convexity} (convexity). Let $\Sc$ denote its support set and $\mathcal{P} := \{i: f(x,\xi^i) < 0\}$. Then $\Sc \cap \mathcal{P} = \emptyset$.
}
\end{lem}
\FR{Intuitively, we would conjecture that the set of active scenarios (i.e., at least one row in $f(x,\xi^i) \le 0$ is equality) is the support set, i.e., $\Sc \cup \mathcal{P} = \mathcal{N}$. Unfortunately, this conjecture is not correct. A counterexample is shown in Figure \ref{fig:primal-example}, in which $\mathcal{P} = \emptyset$ and $\mathcal{N} = \{1,2\}$ but $\Sc = \{2\}$.} \FR{Lemma \ref{lem:inactive_are_not_support_scenario} completely relies on the information from the original (primal) scenario problem. Theorem \ref{thm:support_scenario_multiplier} shows that stronger results are possible by utilizing the optimal dual solution $\{\mu^{i,*}\}_{i=1}^{N}$.}
% \vspace{-0.4cm}
\begin{figure}[htbp]
  \centering
  \includegraphics[width=\linewidth]{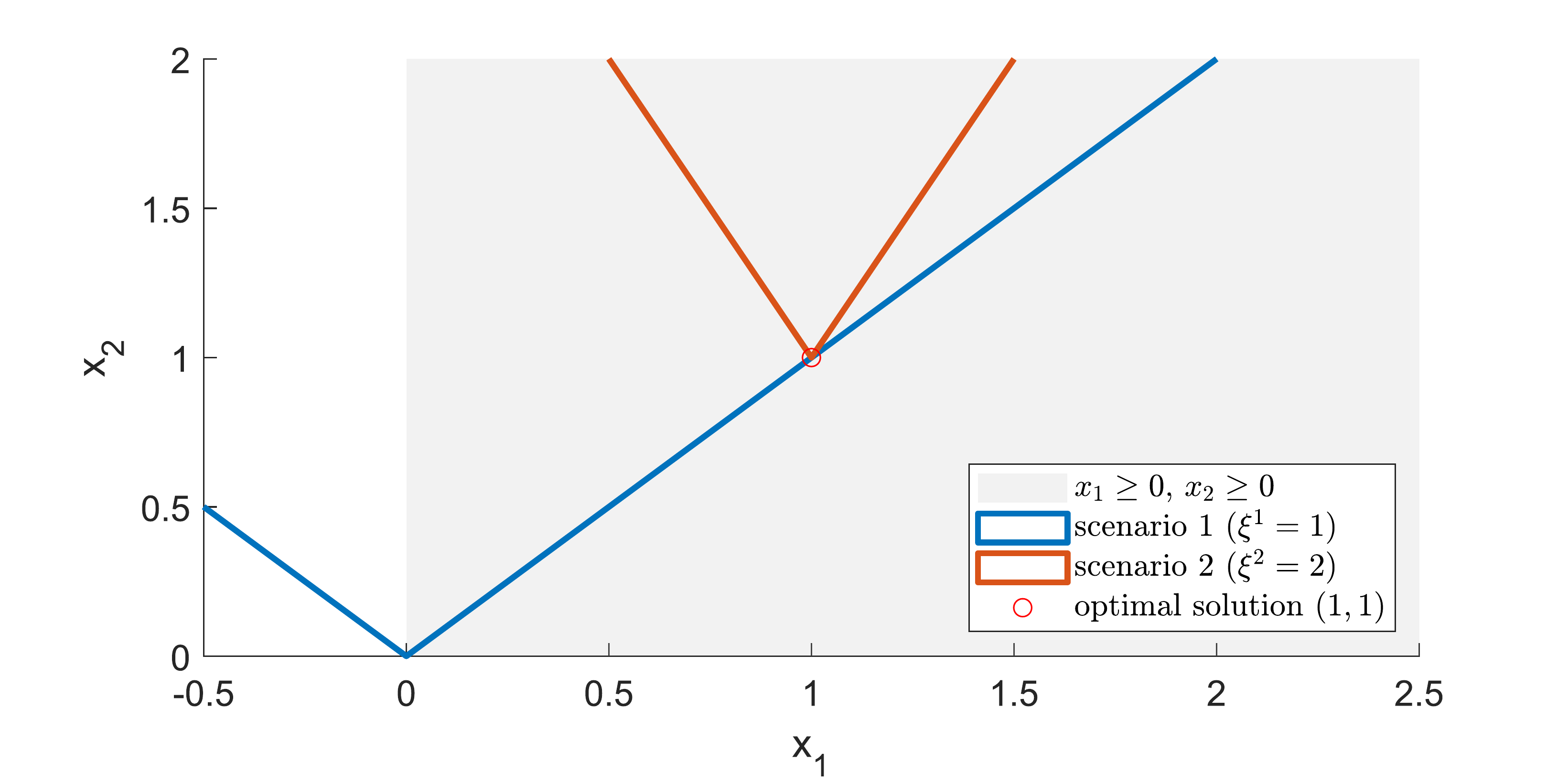}
  \caption{\FR{An illustrative example for Lemma \ref{lem:inactive_are_not_support_scenario} and Theorem \ref{thm:support_scenario_multiplier}. Constraints of two scenarios are visualized in the figure. The optimal solution is $(1,1)$, which is solely determined by scenario $2$, thus the support set $\Sc$ is $\{2\}$. The optimal dual solution is $\mu_{11} = 0, \mu_{12} \ge 0$, $\mu_{21} = 0.5-0.75 \mu_{12}$ and $\mu_{22} = 0.5-0.25 \mu_{12}$. More details are in Appendix \ref{sub:illustrative_example_in_figure_fig:primal-example}.}}
  \label{fig:primal-example}
\end{figure}
% \vspace{-0.4cm}
\begin{thm}
\label{thm:support_scenario_multiplier}
\FR{Consider a \emph{non-degenerate} scenario problem $\SP(\N)$ under Assumptions \ref{ass:feasibility_uniqueness} (feasibility) and \ref{ass:convexity} (convexity). Let $\mu^* := \{\mu^{i,*}\}_{i=1}^{N}$ denote an optimal dual solution (may not be unique) of $\SP(\N)$ and define $\mathcal{M}(\mu^*) := \{i \in \mathcal{N}: \|\mu^{i,*}\| > 0\}$.
(1) If $\xi^j$ is a support scenario, then $\|\mu^{j,*}\| > 0$. In other words, $\Sc \subseteq \mathcal{M}(\mu^*)$ for any optimal dual solution $\mu^*$. 
(2) If $\xi^j$ is \emph{not} a support scenario ($j \notin \Sc$), then there exists an optimal dual solution $\mu^{j,*} \in \mathbf{R}_+^{m}$ with $\|\mu^{j,*}\| = 0$. (3) There exists an optimal dual solution $\mu^\star$, such that $\mathcal{M}(\mu^\star) = \Sc$. (4) If the optimal dual solution $\mu^*$ is unique, then $\mathcal{M}(\mu^*) = \Sc$.
}
% \textcolor{blue}{(1) If $\xi^i$ is a support scenario ($i \in \Sc$), then $\|\mu^{i,*}\| > 0$; (2) If $\|\mu^{i,*}\| = 0$, then $\xi^i$ is not a support scenario ($i \notin \Sc$); (3) If $\xi^i$ is not a support scenario ($i \notin \Sc$), then $\|\mu^{i,*}\| = 0$; (4) If  $\|\mu^{i,*}\| > 0$, then $\xi^i$ is a support scenario ($i \in \Sc$). }
\end{thm}
\begin{rem}[The sparsest dual solution $\mu^{\star}$]
\label{rem:sparsest_dual_sol}
\FR{In practice, the challenge of applying Theorem \ref{thm:support_scenario_multiplier} is the non-uniqueness of optimal dual solutions. Ideally, we would like to find the dual solution $\mu^\star$ that satisfies $\mathcal{M}(\mu^\star) = \mathcal{S}$.
Notice this solution is the \emph{sparsest}\footnote{\FR{Consider two dual solutions to the scenario problem in Fig. \ref{fig:primal-example}: (1) $(\mu_{11}^{*},\mu_{12}^{*}) = (0,2/3)$ and $(\mu_{21}^{*},\mu_{22}^{*}) = (0,1/3)$; (2)  $(\mu_{11}^{\star},\mu_{12}^{\star}) = (0,0)$ and $(\mu_{21}^{\star},\mu_{22}^{\star}) = (1/2,1/2)$. Although both solutions have two non-zero entries, $\mu^\star$ is the sparsest solution as $\mathcal{M}(\mu^\star) = \{2\}$ while $\mathcal{M}(\mu^*) = \{1,2\}$.}} dual solution since all other solutions lead to larger set $\mathcal{M}(\mu^*) \supseteq \mathcal{S}$.
A possible approach is to directly solve the dual of $\SP(\N)$ with \emph{regularizers}, instead of solving the primal scenario problem $\SP(\N)$. The regularization terms are mainly for the purpose of getting sparse solutions. This is part of our ongoing works.
}
\end{rem}
\FR{In numerical simulations, the dual solution that an optimization solver returns depends on the choice of algorithm and implementation details.
For example, the scenario problem in Figure \ref{fig:primal-example} is non-degenerate\footnote{\FR{The non-degeneracy here is \emph{different} from the same term used in optimization theory. For example, primal degeneracy in optimization theory often refers to the existence of multiple dual solutions. Unless specified, the non-degeneracy in this paper is referred to Definition \ref{defn:degenerate-convex}.}} according to Definition \ref{defn:degenerate-convex}, but there is an infinite number of optimal dual solutions as calculated in Appendix \ref{sub:illustrative_example_in_figure_fig:primal-example}.
We used MOSEK v9.2.9 and GUROBI v9.0.0 to solve the dual of the example in Figure \ref{fig:primal-example}. MOSEK returns solution $\mu^1 = (0, 0.3692)$ and $\mu^2=(0.2231,0.4077)$, thus $\mathcal{M} = \{1,2\}$.  GUROBI returns solution $\mu^1=(0,0)$ and $\mu^2=(0.5,0.5)$, which is the \emph{sparsest} solution we prefer.}

\FR{When the sparsest dual solution is not readily available, we can use Algorithm \ref{alg:find-support-scenario-dual} to identify the support scenarios. Algorithm \ref{alg:find-support-scenario-dual} is based on Definition \ref{defn:support_scenario}, i.e., checking scenarios in the set $\mathcal{M}(\mu^*)$ one by one.}
\begin{algorithm}[H]
\begin{algorithmic}[1]
\STATE Compute the primal and dual solutions $x_{\mathcal{N}}^*$ and $\mu^{i,*}$ ($i=1,2,\cdots,N$) by solving $\SP(\N)$
\STATE Let $\mathcal{M} = \{i \in \mathcal{N}: \|\mu^{i,*}\| > 0\}$. Set $\mathcal{S} \leftarrow \emptyset$.
\FOR{$i \in \mathcal{M}$}
\STATE Solve $\text{SP}_{\mathcal{M}-i}$ and compute $x_{\mathcal{M}-i}^*$
\IF{$c^\intercal x_{\mathcal{M}-i}^* < c^\intercal x_{\mathcal{N}}^* (= c^\intercal x_{\mathcal{M}}^*)$} 
  \STATE $ \mathcal{S} \leftarrow \mathcal{S} + i$
\ENDIF
\ENDFOR
\end{algorithmic}
\caption{Finding Support Scenarios Using Dual Variables}
\label{alg:find-support-scenario-dual}
\end{algorithm}
% \FR{When the dual solution is unique, the last part in Theorem \ref{thm:support_scenario_multiplier} shows that there is no need to use Algorithm \ref{alg:find-support-scenario-dual}. In this case, however, there is no need to compute the dual variables since the essential set is the set of active scenarios.}
% \begin{cor}
% \label{cor:convex_SP_unique_dual}
% \FR{Consider a \emph{non-degenerate} scenario problem $\SP(\N)$ under Assumptions \ref{ass:feasibility_uniqueness} (feasibility) and \ref{ass:convexity} (convexity), and let $\mathcal{Q}$ denote the set of active scenarios (i.e., at least one row in $f(x,\xi^i) \le 0$ is equality). If the dual solution is unique, then $\mathcal{S} = \mathcal{Q}$.}
% \end{cor}
\begin{rem}[Finding the Essential Set For Convex Problems]
\label{rem:find_support_scenarios_convex}
\FR{If the scenario problem is known to have a unique dual solution, then there is no need to use Algorithm \ref{alg:find-support-scenario-dual}. In practice, it is often difficult to know the uniqueness of the solution before solving the problem. So we need to rely on solvers or slightly modify the scenario problem to get the sparsest solution (Theorem \ref{thm:support_scenario_multiplier} and Remark \ref{rem:sparsest_dual_sol}). In the worst-case scenario, we compute an optimal dual solution then apply Algorithm \ref{alg:find-support-scenario-dual}, which requires us to solve the scenario problem $|\mathcal{M}|$ times. In many cases (especially in power system applications, e.g., \cite{modarresi_scenario-based_2018}), it is observed that the support scenarios are only a small subset of all $N$ scenarios, i.e., $|\mathcal{M}| \approx |\Sc| \ll |\N|$. Comparing with Algorithms \ref{alg:find-irreducible-set} or \ref{alg:find-support-set}, Algorithm \ref{alg:find-support-scenario-dual} only requires solving $\sim |\Sc|$ scenario problems, which is much more efficient since $|\Sc| \ll |\N|$.}
\end{rem}
% Algorithm \ref{alg:find-support-scenario-dual} is also utilized in finding essential sets for the non-convex two-stage scenario problems.
% subsection convex_scenario_problems (end)
% \vspace{-0.7cm}
\subsection{Two-stage Scenario Problems} % (fold)
\label{sub:two_stage_scenario_problems}
Section \ref{sub:non_convex_scenario_problems} shows that searching for essential sets can be relatively easier when a scenario problem is non-degenerate. However, finding a support set or irreducible set still requires solving $N$ non-convex problems. Motivated by SCUC, we show that more efficient algorithms are possible by exploiting the structure of specific problems. We study the following two-stage scenario problem in this subsection.
\begin{subequations}
\label{opt:two-stage-scenario-problem}
\begin{align}
\min_{y \in \mathcal{Y}}  c_y^\intercal y + \min_{\substack{x \in \mathcal{X}\\(x,y) \in \mathcal{H}}}~& c_x^\intercal x \\
\text{s.t.}~& x \in \cap_{i=1}^N \mathcal{U}_{i}
\end{align}
\end{subequations}
Constraints on the first-stage variables $y$ and the second-stage variables $x$ are denoted by $y \in \mathcal{Y}$ and $x \in \mathcal{X}$, respectively. Constraint $(x,y) \in \mathcal{H}$ represents the constraints coupling variables $x$ and $y$ in both stages. Set $\mathcal{U}_i$ stands for the constraints corresponding to the $i$th scenario $\xi^i$.

Problem \eqref{opt:two-stage-scenario-problem} is an abstract form of s-SCUC in Section \ref{sec:security_constrained_unit_commitment}.
Two key features of the two-stage scenario problem are: (1) the non-convexity \emph{only} comes from constraints $y \in \mathcal{Y}$ (e.g., binary variables in SCUC), all other constraints ($\mathcal{X}, \mathcal{H}, \mathcal{U}_i$) are convex; (2) uncertainties only exist in the second stage.

Let $(x^*, y^*)$ be a (possibly local) optimal solution that algorithm $\mathbb{A}$ returns. Given $y = y^*$, the second stage problem is convex by setting:
\begin{subequations}
\label{opt:second_stage_scenario_problem}
\begin{align}
\min_{\substack{x \in \mathcal{X}\\(x,y^*) \in \mathcal{H}}}~& c_x^\intercal x \\
\text{s.t.}~& x \in \cap_{i=1}^N \mathcal{U}_{i}
\end{align}
\end{subequations}
% \begin{align}
% \min_{\substack{x \in \mathcal{X}\\(x,y^*) \in \mathcal{H}}}~& c_x^\intercal x,~\text{s.t.}~x \in \cap_{i=1}^N \mathcal{U}_{i} \label{opt:second_stage_scenario_problem}
% \end{align}
% WLOG, we assume that $x^*$ is the global optimal solution to the 2nd stage problem.
\begin{lem}
\label{lem:two_stage_support}
(1) Let $\hat{\Sc}$ represent the set of support scenarios of \eqref{opt:second_stage_scenario_problem} and $\Sc$ denote the support set for the two-stage problem \eqref{opt:two-stage-scenario-problem}, then $\hat{\Sc} \subseteq \Sc$; (2) If $\hat{\Sc}$ is invariant for \eqref{opt:two-stage-scenario-problem}, i.e., $\opt_\mathbb{A}(\hat{\Sc}) = \opt_\mathbb{A}(\N)$, then the two-stage scenario problem $\SP_{\mathbb{A}}(\N)$ is non-degenerate, \FR{therefore $\mathcal{S} = \hat{\mathcal{S}}$}.
\end{lem}
Corollary \ref{cor:nondegenerate_uniqueness} and Theorem \ref{thm:nondegenerate_unique_essential_irreducible} demonstrate many nice properties of non-degenerate scenario problems.
Lemma \ref{lem:two_stage_support} gives a criteria of checking if the two-stage problem (e.g., s-SCUC) is non-degenerate. This lemma lays the foundation of Algorithm \ref{alg:find-support-subsample-structural} to search for essential sets of \eqref{opt:two-stage-scenario-problem}. The main idea of Algorithm \ref{alg:find-support-subsample-structural} is to first find the support scenarios of the second-stage problem \eqref{opt:second_stage_scenario_problem}, then verify if $\SP(\N)$ is degenerate using Lemma \ref{lem:two_stage_support}. In Section \ref{sub:finding_essential_sets_for_s_scuc}, it turns out that s-SCUC is non-degenerate in $96\%$ of cases, thus Algorithm \ref{alg:find-support-subsample-structural} could obtain essential sets of s-SCUC (in Section \ref{sub:finding_essential_sets_for_s_scuc}) in a much shorter time.

% Generalized Benders Algorithm
% subsection two_stage_scenario_problems (end)
% section structural_properties_of_general_scenario_problems (end)

\section{Security-Constrained Unit Commitment with Probabilistic Guarantees} % (fold)
\label{sec:security_constrained_unit_commitment}
\subsection{Nomenclature} % (fold)
\label{sub:nomenclature}
The number of loads, generators, wind farms, transmission lines, contingencies, and snapshots are denoted by $n_d$, $n_g$, $n_w$, $n_l$, $n_k$ and $n_t$, respectively.
\begin{labeling}{$\iota \in \{t+1,\cdots,n_t\}$}
\item[$k \in \{0,1,\cdots,n_k\}$] contingency index
\item[$t \in \{0,1,\cdots,n_t\}$] time (snapshot) index
\item[$\iota \in \{t+1,\cdots,n_t\}$] additional time (snapshot) index in constraints \eqref{opt:det-SCUC-minon-time} and \eqref{opt:det-SCUC-minoff-time}
\end{labeling}
Binary decision variables (at time $t$):
\begin{labeling}{$z^{t} \in \{0,1\}^{n_g}$}
\item[$z^{t} \in \{0,1\}^{n_g}$] generator on/off states (commitment)
\item[$u^t \in \{0,1\}^{n_g}$] generator $i$ is on if $u_i^t = 1$
\item[$v^t \in \{0,1\}^{n_g}$] generator $i$ is off if $v_i^t = 1$ 
\end{labeling}
Continuous decision variables (at time $t$, contingency $k$):
\begin{labeling}{$g^{t,k} \in \mathbf{R}^{n_g}$}
\item[$g^{t,k} \in \mathbf{R}^{n_g}$] generation output
% $l^{t,k} \in \mathbf{R}^{n_d}$ & load shedding  \\
% $s^{t,k} \in \mathbf{R}^{n_d}$ & wind spillage/curtailment \\
\item[$r^{t} \in \mathbf{R}^{n_g}$] reserve 
\end{labeling}
Parameters and constants:
\begin{labeling}{$a^k \in \{0,1\}^{n_g}$}
% \item[$N^s, N^n$] Subset of buses with/without energy storage
\item[$a^k \in \{0,1\}^{n_g}$] generator availability in contingency $k$ 
\item[$\alpha_k \in \mathbf{R}_+$] weight of contingency $k$
\item[$c_g \in \mathbf{R}^{n_g}$] generation costs
\item[$c_z \in \mathbf{R}^{n_g}$] no load cost 
\item[$c_r \in \mathbf{R}^{n_g}$] reserve costs
\item[$c_u \in \mathbf{R}^{n_g}$] startup cost 
\item[$c_v \in \mathbf{R}^{n_g}$] shutdown cost 
\item[$\hat{d}^{t} \in \mathbf{R}^{n_d}$] load forecast (time $t$)
\item[$\tilde{d}^{t} \in \mathbf{R}^{n_d}$] load forecast error (time $t$)
\item[$\hat{w}^{t} \in \mathbf{R}^{n_w}$] wind forecast (time $t$) 
\item[$\tilde{w}^{t} \in \mathbf{R}^{n_w}$] wind forecast error (time $t$) 
\item[$\overline{g} \in \mathbf{R}^{n_g}$] generation upper bounds
\item[$\underline{g} \in \mathbf{R}^{n_g}$] generation lower bounds
\item[$\overline{\gamma} \in \mathbf{R}^{n_g}$] ramping upper bounds
\item[$\underline{\gamma} \in \mathbf{R}^{n_g}$] ramping lower bounds 
\item[$\underline{u}_i \in \mathbf{R}_+$] minimum on time for generator $i$
\item[$\underline{v}_i \in \mathbf{R}_+$] minimum off time for generator $i$
\end{labeling}
\subsection{Deterministic SCUC (d-SCUC)} % (fold)
\label{sub:deterministic_scuc}
Deterministic security-constrained unit commitment (d-SCUC) \eqref{opt:det-SCUC} seeks optimal commitment and startup/shutdown decisions $(z^t,u^t, v^t)$, generation and reserve schedules $(g^{t,k}, r^t)$ for a horizon of time steps, typically $24$ hours. Security constraints ensures the reliability of the power system after an unexpected event occurs.
\begin{subequations}
\label{opt:det-SCUC}
\begin{align}
\min_{z,u,v,g,r}~ & \sum_{t=1}^{n_t} \Big(  c_z^\intercal z^t + c_u^\intercal u^t + c_v^\intercal v^t + c_r^\intercal r^t + \sum_{k=0}^{n_k} \alpha_k c_g^\intercal g^{t,k} \Big) \label{opt:det-SCUC-obj} \\
\text{s.t.}~& \mathbf{1}^\intercal g^{t,k} + \mathbf{1}^\intercal \hat{w}^t  \ge \mathbf{1}^\intercal \hat{d}^t \label{opt:det-SCUC-balance}\\
& \underline{f} \le H_g^{t,k} g^{t,k} + H_w^{t,k} w^{t,k} - H_d^{t,k} d^{t,k} \le \overline{f} \label{opt:det-SCUC-line} \\
& a^k \circ \underline{\gamma} \le g^{t,k} - g^{t-1,k} \le a^k \circ \overline{\gamma}  \label{opt:det-SCUC-ramp} \\
& a^k \circ (g^{t,0} - r^t) \le g^{t,k} \le a^k \circ (g^{t,0} + r^t) \label{opt:det-SCUC-contingency} \\
% & a^k \circ  \underline{g} \circ  z^t \le g^{t,k} \le a^k \circ  \overline{g} \circ  z^t \label{opt:det-SCUC-contgencap-redundant} \\
& \hspace{109pt} k \in [0,n_k], t \in [1,n_t] \nonumber \\
& \underline{g} \circ  z^t \le g^{t,0} \le \overline{g} \circ  z^t \label{opt:det-SCUC-gencap}\\
& \underline{g} \circ  z^t \le g^{t,0} - r^t \le g^{t,0} + r^t \le \overline{g} \circ  z^t \label{opt:det-SCUC-resgencap}\\
& z^{t-1} - z^t + u^t \ge 0 \label{opt:det-SCUC-startup} \\
& z^t - z^{t-1} + v^t \ge 0 \label{opt:det-SCUC-shutdown} \\
& \hspace{157pt} t \in [1,n_t] \nonumber \\
& z_i^{t} - z_i^{t-1} \le z_i^\iota, ~ i \in [1,n_g] \label{opt:det-SCUC-minon-time} \\
& \hspace{23pt} \iota \in [t+1,\min\{t+\underline{u}_i-1,n_t\}], t \in [2,n_t] \nonumber \\
& z_i^{t-1} - z_i^{t} \le 1- z_i^\iota, ~ i \in [1,n_g] \label{opt:det-SCUC-minoff-time}\\
& \hspace{23pt} \iota \in [t+1,\min\{t+\underline{v}_i-1,n_t\}], t \in [2,n_t] \nonumber
\end{align}
\end{subequations}
The objective of \eqref{opt:det-SCUC} is to minimize total operation costs, including no-load costs $c_z^\intercal z^t$, startup costs $c_u^\intercal u^t$, shutdown costs $c_v^\intercal v^t$, generation costs $c_g^\intercal g^{t,k}$ and reserve costs $c_r^\intercal s^t$. Security constraints ensure: enough supply to meet demand \eqref{opt:det-SCUC-balance}, transmission line flow within limits \eqref{opt:det-SCUC-line}, generation levels within ramping limits \eqref{opt:det-SCUC-ramp} and capacity limits \eqref{opt:det-SCUC-gencap} in any contingency $k$. Constraints \eqref{opt:det-SCUC-contingency} and \eqref{opt:det-SCUC-resgencap} are about the relationship between generation and reserve in any contingency $k$. Constraints \eqref{opt:det-SCUC-startup}-\eqref{opt:det-SCUC-shutdown} are the logistic constraints about commitment status, startup and shutdown decisions. Minimum on/off time constraints for all generators are in \eqref{opt:det-SCUC-minon-time}-\eqref{opt:det-SCUC-minoff-time}. Constraints \eqref{opt:det-SCUC-ramp}-\eqref{opt:det-SCUC-resgencap} also guarantee the consistency of generation levels $g^{t,k}$ with commitment decisions $z^t$ and generator availability $a^k$ in contingency $k$ \cite{geng_chance-constrained_2019}. 

To reveal the structure of d-SCUC, we define the sets below:
\begin{subequations}
\begin{align}
\mathcal{B} &:= \big\{ (z,u,v): \eqref{opt:det-SCUC-startup}, \eqref{opt:det-SCUC-shutdown}, \eqref{opt:det-SCUC-minon-time}, \eqref{opt:det-SCUC-minoff-time} \big\}  \\
\mathcal{C} &:= \big\{ (g,r):  \eqref{opt:det-SCUC-balance}, \eqref{opt:det-SCUC-line}, \eqref{opt:det-SCUC-ramp}, \eqref{opt:det-SCUC-contingency} \big\}  \\
\mathcal{H} &:= \big\{ (z,g,r): \eqref{opt:det-SCUC-gencap}, \eqref{opt:det-SCUC-resgencap}  \big\}
% \mathcal{U} &:= \big\{(g): \eqref{opt:c-SCUC-U}  \big\} 
\end{align}
\end{subequations}
Sets $\mathcal{B}$ and $\mathcal{C}$ stand for the \emph{deterministic} constraints for binary and continuous variables, respectively. 
Set $\mathcal{H}$ represents the hybrid constraints related with both continuous and binary variables.
Then d-SCUC can be succinctly represented as 
\begin{align*}
\min_{z,u,v,g,r}~ & \eqref{opt:det-SCUC-obj}  \\
\text{s.t.}~ &  (z,u,v) \in \mathcal{B},~(g,r) \in \mathcal{C},~(z,g,r) \in \mathcal{H}
\end{align*}
% \begin{align*}
% \min_{z,u,v,g,r}~ &\eqref{opt:det-SCUC-obj},~\text{s.t.}~(z,u,v) \in \mathcal{B},~(g,r) \in \mathcal{C},~(z,g,r) \in \mathcal{H}
% \end{align*}

\subsection{Chance-constrained SCUC (c-SCUC)} % (fold)
\label{sub:chance_constrained_scuc}
The d-SCUC formulation utilizes the expected wind generation and load forecast, it does not take the uncertainties from wind and load into consideration.
\FR{Various formulations of chance-constrained SCUC (c-SCUC) have been proposed, e.g., \cite{ozturk_solution_2004,wang_chance-constrained_2012}.}
The c-SCUC formulations explicitly guarantee the system security with a tunable level of risk $\epsilon$ with respect to uncertainties. Instead of using expected load $\hat{d}^t$ as in \eqref{opt:det-SCUC}, we consider loads $d^t$ as forecast $\hat{d}^t$ plus a random forecast error $\tilde{d}^t$ (i.e., $d^t = \hat{d}^t + \tilde{d}^t$).
% \begin{subequations}
% \label{opt:c-SCUC-U}
% \begin{align}
% \mathbb{P}_{\tilde{w}\times\tilde{d}}\Big( \mathbf{1}^\intercal g^{t,k} + \mathbf{1}^\intercal (\hat{w}^t + \tilde{w}^t) \ge \mathbf{1}^\intercal (\hat{d}^t + \tilde{d}^t ), \label{opt:c-SCUC-balance-U} \\
% \underline{f} \le H_g^{t,k} g^{t,k} + H_w^{t,k} (\hat{w}^{t}+\tilde{w}^{t}) - H_d^{t,k} (\hat{d}^{t}+\tilde{d}^{t}) \le \overline{f}, \label{opt:c-SCUC-line-U} \\
% \hspace{1.5cm} k \in [0,n_k], t \in [1,n_t] \Big ) \ge 1 - \epsilon \nonumber
% \end{align}
% \end{subequations}
\begin{multline} \label{opt:chance-constrained-in-SCUC}
\mathbb{P}_{\tilde{w}\times\tilde{d}}\Big( \mathbf{1}^\intercal g^{t,k} + \mathbf{1}^\intercal (\hat{w}^t + \tilde{w}^t) \ge \mathbf{1}^\intercal (\hat{d}^t + \tilde{d}^t ), \\
\underline{f} \le H_g^{t,k} g^{t,k} + H_w^{t,k} (\hat{w}^{t}+\tilde{w}^{t}) - H_d^{t,k} (\hat{d}^{t}+\tilde{d}^{t}) \le \overline{f}, \\
k \in [0,n_k], t \in [1,n_t] \Big ) \ge 1 - \epsilon
\end{multline}
We define set $\mathcal{U} := \{g: \eqref{opt:chance-constrained-in-SCUC}\}$ to represent all constraints related with uncertainties. Then the formulation of chance-constrained Security-constrained Unit Commitment (c-SCUC) is presented below. Comparing with d-SCUC, the only difference of c-SCUC is the addition of the chance constraint \eqref{opt:chance-constrained-in-SCUC}. The chance constraint guarantees there will be enough supply to meet the net demand with probability no less than $1 - \epsilon$.
\begin{align*}
\min_{z,u,v,g,r}~ & \eqref{opt:det-SCUC-obj}  \\
\text{s.t.}~ & (z,u,v) \in \mathcal{B},~(g,r) \in \mathcal{C},~(z,g,r) \in \mathcal{H}\\
& \mathbb{P}\big ( g \in \mathcal{U}  \big) \ge 1 - \epsilon
\end{align*}
% \begin{subequations}
% \label{opt:c-SCUC-U}
% \begin{align}
%  \nonumber
% \end{align}
% \end{subequations}

% \begin{align}
%   \min~ &  \eqref{opt:det-SCUC-obj} \nonumber \\
% \text{s.t.}~& \eqref{opt:det-SCUC-balance}\eqref{opt:det-SCUC-line}\eqref{opt:det-SCUC-ramp}\eqref{opt:det-SCUC-contingency}\eqref{opt:det-SCUC-gencap}\eqref{opt:det-SCUC-resgencap}\eqref{opt:det-SCUC-startup}\eqref{opt:det-SCUC-shutdown}\eqref{opt:det-SCUC-minon-time}\eqref{opt:det-SCUC-minoff-time} \nonumber \\
% & \eqref{opt:c-SCUC-U} \nonumber 
% \end{align}

\subsection{Scenario-based SCUC (s-SCUC)} % (fold)
\label{sub:scenario_based_scuc}
The scenario approach was mainly targeted at convex problems (see Assumption \ref{ass:convexity}), whereas SCUC is non-convex by nature due to on/off commitment decisions. Consequently, the scenario approach was considered \emph{not applicable} for c-SCUC. An extended version of the scenario approach was proposed recently in \cite{campi_general_2018}, which makes it applicable for non-convex problems such as SCUC. Another related paper is in \cite{marley_ac-qp_2016}, which applies Theorem \ref{thm:scenario_theory_nonconvex_posterior} on the AC optimal power flow problem and obtain similar posterior guarantees.

Using the scenario approach, c-SCUC is \FR{approximated by} the scenario-based SCUC (s-SCUC) problem below:
\begin{align*}
\min_{(z,u,v) \in \mathcal{B}}~& \sum_{t=1}^{n_t} \Big( c_z^\intercal z^t + c_u^\intercal u^t + c_v^\intercal v^t \Big) + \\
& \min_{(z,g,r) \in \mathcal{H}} \sum_{t=1}^{n_t} \Big( c_r^\intercal r^t + \sum_{k=0}^{n_k} \alpha_k c_g^\intercal g^{t,k} \Big) \nonumber \\
 & \hspace{1cm}\text{s.t.}~ (g,r) \in \mathcal{C} \\
 & \hspace{1.5cm} g \in \cap_{i=1}^N \mathcal{U}_i \nonumber 
\end{align*}
\begin{rem}
\label{rem:structural_prop_c-SCUC}
SCUC is a two-stage optimization problem by nature, it has the following nice properties. Firstly, the non-convexity \emph{only} exists in the first stage, i.e., $y \in \mathcal{Y}$. Given a first-stage solution $y$, the second stage is a simple linear program. Secondly, uncertainties come from renewables in the operation stage (only in the second stage).
% The abstract form \eqref{opt:two-stage-scenario-problem} of of s-SCUC highlights these two .
Based on the nice structural properties above, Section \ref{sub:two_stage_scenario_problems} shows that we are able to track down essential sets by solving two MILPs and $\sim |\Sc|$ linear programs.
\end{rem}

\subsection{Degeneracy of s-SCUC} % (fold)
\label{sub:s_scuc_is_degenerate}
This section presents an example to show that s-SCUC could be degenerate in many cases, which violates Assumption \ref{ass:non-degeneracy}. Therefore almost all results of the classical scenario approach are not applicable. For s-SCUC, theoretical guarantees are \emph{only possible} through the non-convex scenario approach in Section \ref{sub:the_scenario_approach_for_non_convex_problems}.

We use a 3-bus system to illustrate the degeneracy of s-SCUC. Configurations of the 3-bus system are in \cite{geng_computing_2019}. In order to visualize the feasible region of s-SCUC, we simplify the problem by (1) only considering one snapshot ($n_t=1$) and ignoring initial status (thus no $u,v$ variables); (2) removing reserve constraints (no $r$ variables). By doing this, there are only four decision variables left: $z_1, z_2, g_1,g_2$. The on/off states $z_1,z_2$ can be inferred from values of $g_1$ and $g_2$, therefore the feasible region of the simplified s-SCUC can be visualized on the $(g_1, g_2)$-plane.

% \begin{figure}[htbp]
% 	\centering
% 	\includegraphics[width=\linewidth]{case3-degenerate-example-all.pdf}
% 	\caption{An illustrative example that s-SCUC is degenerate (3-bus system), illustration of the feasible region with constraints of all scenarios ($\mathcal{U}_1$,$\mathcal{U}_2$,$\mathcal{U}_3$).}
% 	\label{fig:case3-degenerate-example-all}
% \end{figure}

% \begin{figure}[htbp]
% 	\centering
% 	\includegraphics[width=\linewidth]{case3-degenerate-example-support.pdf}
% 	\caption{An illustrative example that s-SCUC is degenerate (3-bus system), illustration of the feasible region with only support scenarios ($\mathcal{U}_1$).}
% 	\label{fig:case3-degenerate-example-support}
% \end{figure}

\begin{figure*}[htbp]
  \centering
  \subfloat[Illustration of the feasible region with constraints of all scenarios ($\mathcal{U}_1$,$\mathcal{U}_2$,$\mathcal{U}_3$).]{\includegraphics[width=0.5\linewidth]{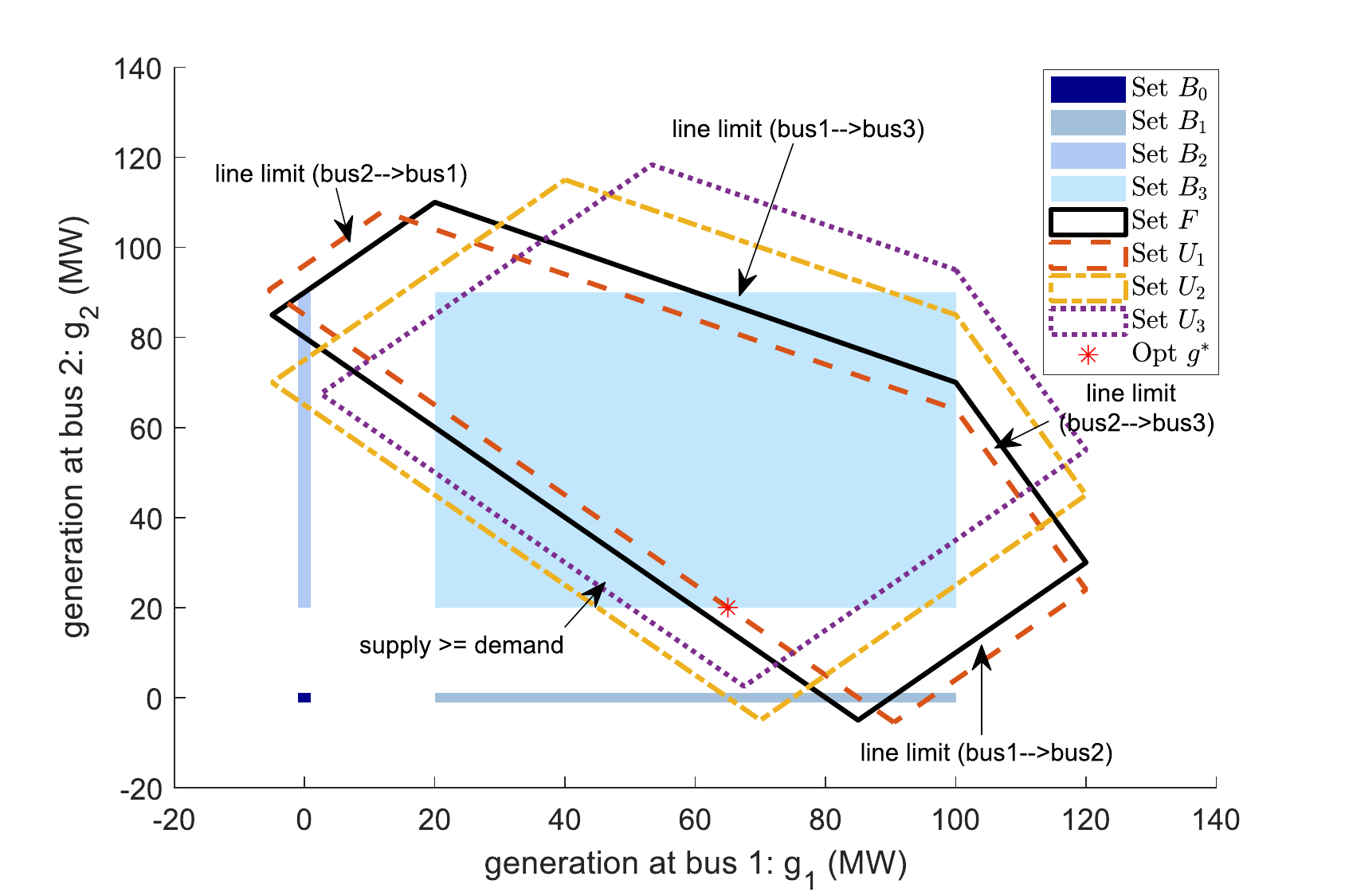}\label{fig:case3-degenerate-example-all}} 
  \subfloat[Illustration of the feasible region with only support scenarios ($\mathcal{U}_1$).]{\includegraphics[width=0.5\linewidth]{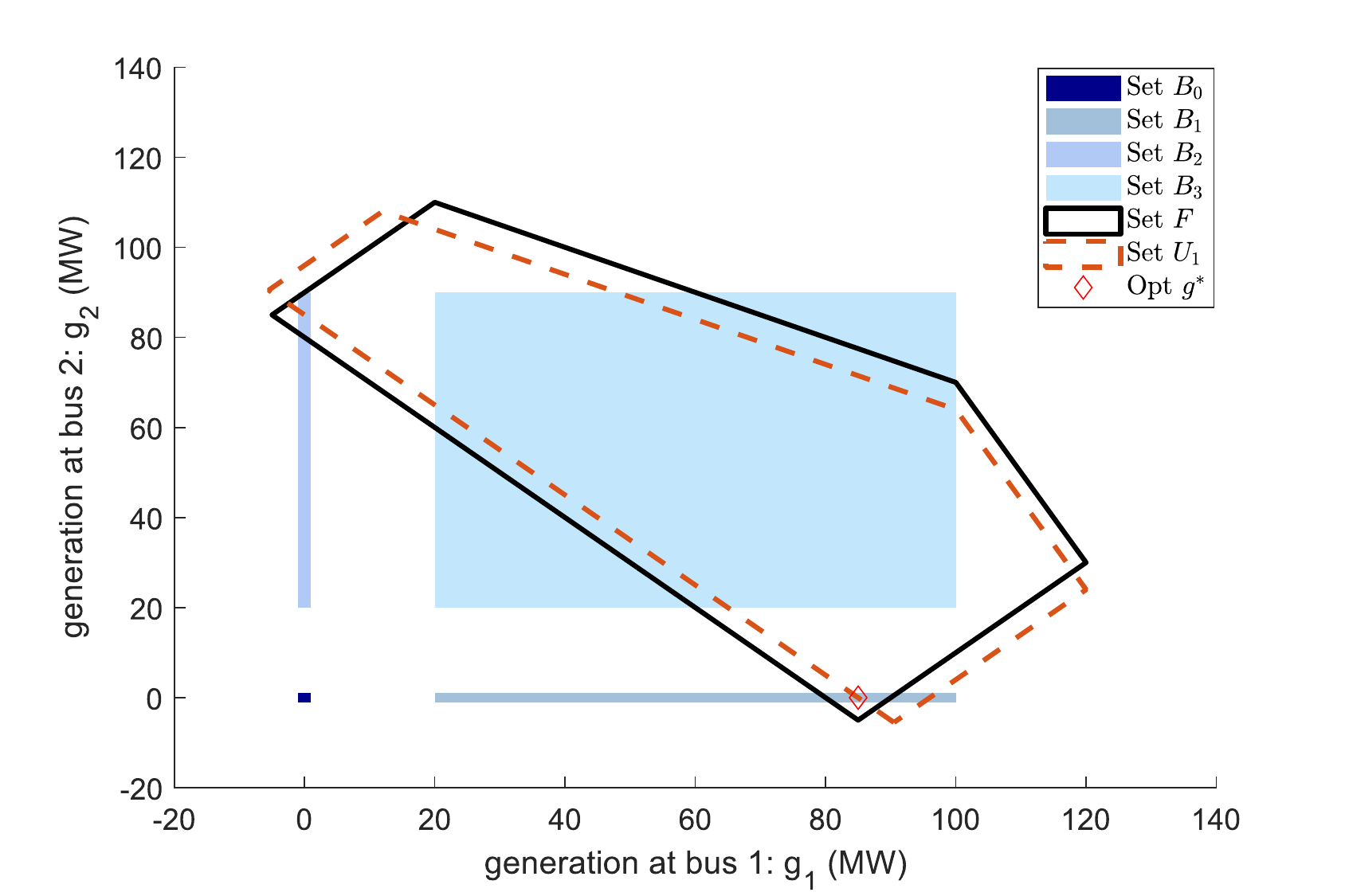}\label{fig:case3-degenerate-example-support}} 
  \caption{An illustrative example that s-SCUC is degenerate (3-bus system)}
  \label{fig:case3-degenerate-example}
\end{figure*}

Using Definition \ref{defn:degenerate-convex}, showing the degeneracy of s-SCUC includes three steps: (1) obtaining the optimal solution to $\SP(\N)$; (2) finding all support scenarios $\Sc$ of $\SP(\N)$; and (3) checking if the optimal solution of $\SP(\N)$ is the same as $\SP(\Sc)$.
Fig. \ref{fig:case3-degenerate-example-all} first visualizes constraints $\mathcal{B}_0\sim \mathcal{B}_3$, which represents the region of 4 possible generator on/off status (e.g., $\mathcal{B}_1: z_1=1, z_2 =0$, $\mathcal{B}_3: z_1=1, z_2 =1$). The black solid lines denote constraints \eqref{opt:det-SCUC-balance}, \eqref{opt:det-SCUC-line} and \eqref{opt:det-SCUC-gencap} using forecast values (d-SCUC). The red, yellow and purple dotted lines are three sets ($\mathcal{U}_1,\mathcal{U}_2,\mathcal{U}_3$) of constraints corresponding to three scenarios.
Given the setting that generator 1 is much cheaper than generator 2, we can easily eyeball the optimal solution with all constraints presented, marked by the red $*$. Next, we observe that removing scenario 1 ($\mathcal{U}_1$, red lines) changes the optimal solution, while removing scenario scenario 2 ($\mathcal{U}_2$, yellow lines) or scenario 3 ($\mathcal{U}_3$, purple lines) makes no difference. Thus scenario 1 is the only support scenario. Finally, we examine the scenario problem with only support scenarios presented. Fig. \ref{fig:case3-degenerate-example-support} shows that the optimal solution becomes the red $\diamond$ with only scenario 1, which is clearly different than the optimal solution in Fig. \ref{fig:case3-degenerate-example-all}. Hence, s-SCUC is a degenerate problem.

% subsection s_scuc_is_degenerate (end)
% section security_constrained_unit_commitment (end)

\section{Numerical Results} % (fold)
\label{sec:case_study}
\subsection{Settings of the 118-bus System} % (fold)
Numerical simulations were conducted on a modified 118-bus, 184-line, 54-generator, 24-hour system \cite{illinois_institute_of_technology_ieee_2004}. Most settings are identical as \cite{illinois_institute_of_technology_ieee_2004}, except 5 wind farms are added to the system as in \cite{pena_extended_2018}. The s-SCUC problems were solved using $64$ GB memory on the Hera server (hera.ece.tamu.edu), provided by Texas A\&M University. The mathematical models for s-SCUC was formulated using YALMIP \cite{lofberg_yalmip_2004} on Matlab R2019a and solved using Gurobi v8.1 \cite{gurobi_optimization_gurobi_2016}.

After obtaining a solution $\opx_{\mathbb{A}}(\N)$ to s-SCUC, Theorem \ref{thm:scenario_theory_nonconvex_posterior} provides an upper bound $\epsilon(N,|\Iv|,\beta)$ on the \emph{actual} violation probability $\mathbb{V}(\opx_{\mathbb{A}}(\N))$. The theoretical guarantee $\epsilon(N,|\Iv|,\beta)$ is referred as posterior $\epsilon$ in the numerical results. The actual violation probability $\mathbb{V}(\opx_{\mathbb{A}}(\N))$ is estimated by the out-of-sample violation probability $\hat{\epsilon}$, using an independent set of $10^6$ scenarios.

To quantify the randomness of the scenario approach, for each sample complexity $N=100,200,\cdots,1000$, we solve the corresponding s-SCUC problems on 10 independent sets of scenarios (i.e., $10$ Monte-Carlo simulations). Results in both Fig. \ref{fig:sc-SCUC-case118iit-fig_obj0eps0_2axis} and \ref{fig:sc-SCUC-case118iit-epsilon-epsilon} show the average, maximum and minimum values in 10 Monte-Carlo simulations. 

% We mainly examine the quality of solution to s-SCUC in the following aspects: (1) economic efficiency, quantified by the total operation cost; (2) system security, quantified by the out-of-sample violation probability $\hat{\epsilon}$, evaluated on an independent set of $10^6$ scenarios; and (3) structural properties, mainly the cardinality of essential sets.

% subsubsection settings (end)
\subsection{Cost vs Security: a trade-off} % (fold)
\label{ssub:cost_security_trade_off}
Fig. \ref{fig:sc-SCUC-case118iit-fig_obj0eps0_2axis} shows the out-of-sample violation probability $\hat{\epsilon}$ and objective value (total cost). The shadowed area shows the max-min values in 10 Monte-Carlo simulations, and the solid line is the average value. It is shown that the system risk level (violation probability) is reduced by $83\%$ (from $\sim 30\%$ to $\sim 5\%$) by $\sim1.1\%$ increase in total system costs. Similar observations were found in \cite{modarresi_scenario-based_2018,geng_data-driven_2019-2,geng_chance-constrained_2019}.
\begin{figure}[htbp]
  \centering
  \includegraphics[width=\linewidth]{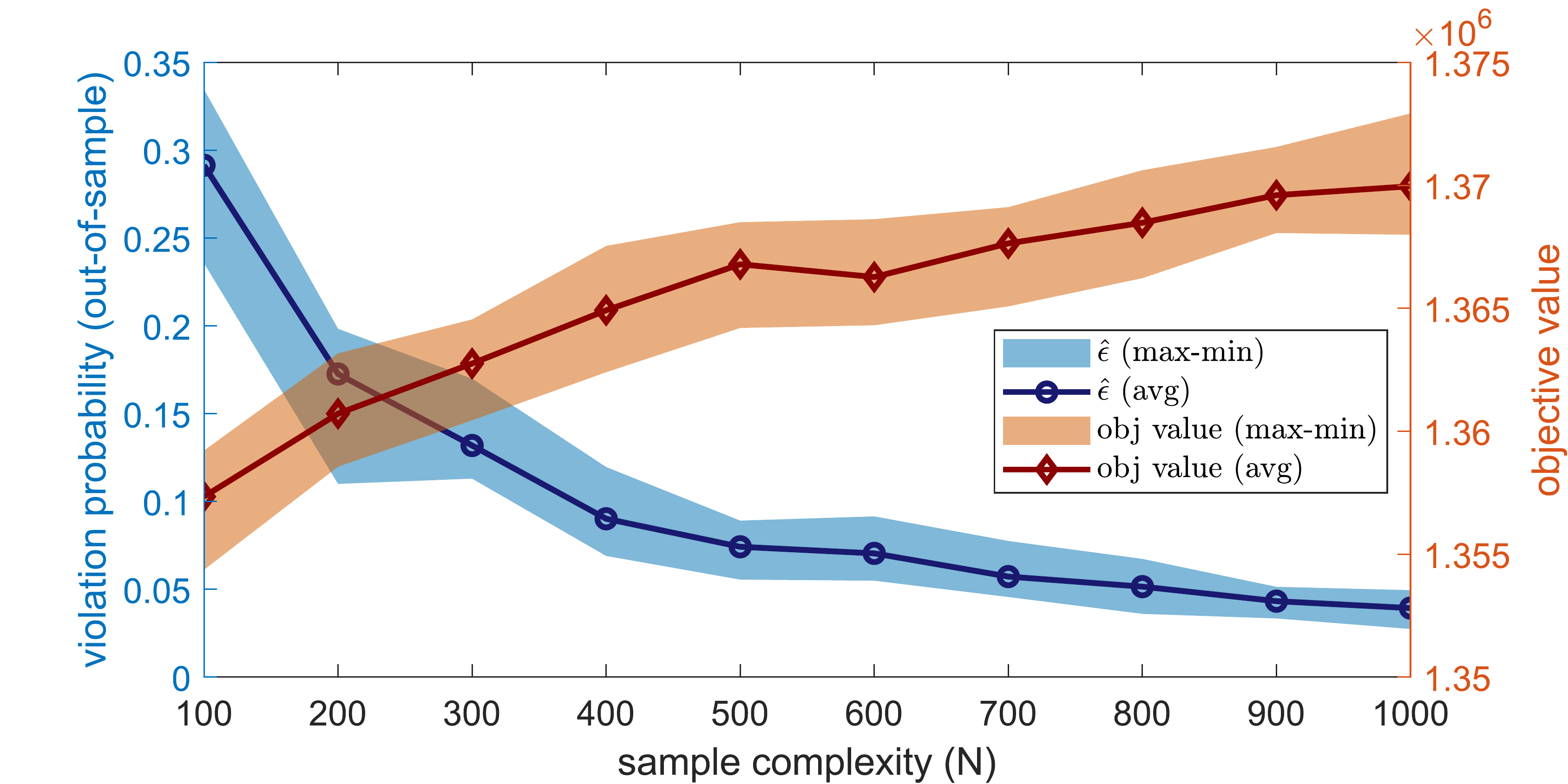}
  \caption{Cost vs Security: a Trade-off.}
  \label{fig:sc-SCUC-case118iit-fig_obj0eps0_2axis}
\end{figure}

% subsubsection cost_security_trade_off (end)

\subsection{Violation Probability} % (fold)
\label{ssub:violation_probability}
Fig. \ref{fig:sc-SCUC-case118iit-epsilon-epsilon} presents the out-of-sample violation probability $\hat{\epsilon}$ and theoretical guarantees (posterior $\epsilon$ provided by Theorem \ref{thm:scenario_theory_nonconvex_posterior}). Since the cardinality of essential sets differ for each scenario problem (Fig. \ref{fig:sc-SCUC-case118iit-support-scenario-number}), the posterior guarantee $\epsilon$ is a band instead of a line. As illustrated in Fig. \ref{fig:sc-SCUC-case118iit-epsilon-epsilon}, the actual violation probability (approximated by $\hat{\epsilon}$) is bounded by the theoretical guarantees. This verifies the correctness of Theorem \ref{thm:scenario_theory_nonconvex_posterior}. The conservative ratio is $2\sim 4$ (e.g., when out-of-sample $\hat{\epsilon}$ is $\sim5\%$, Theorem \ref{thm:scenario_theory_nonconvex_posterior} gives an upper bound $10\% \sim 20\%$).
\begin{figure}[htbp]
  \centering
  \includegraphics[width=\linewidth]{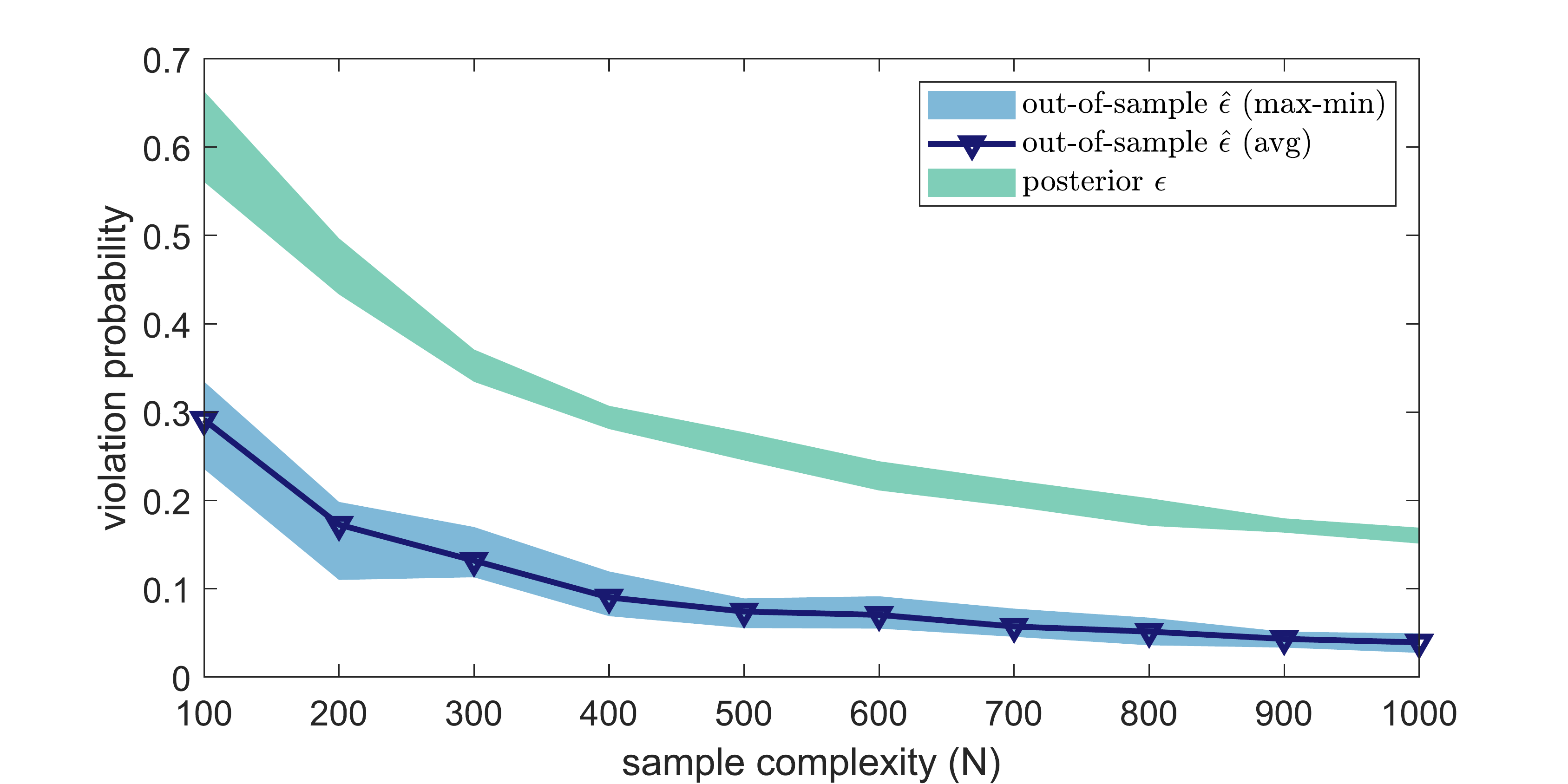}
  \caption{Out-of-sample Violation Probabilities and Theoretical Guarantees.}
  \label{fig:sc-SCUC-case118iit-epsilon-epsilon}
\end{figure}
% subsubsection violation_probability (end)

\subsection{Searching for Essential Sets for s-SCUC} % (fold)
\label{sub:finding_essential_sets_for_s_scuc}
s-SCUC was observed to be non-degenerate in 192 out of 200 simulations\footnote{We conducted 10 simulation for 10 different sample complexities ($100,200,\cdots,1000$) under two different settings: with/without $N-1$ contingencies, both include transmission constraints.}. In other words, in $96\%$ cases, we are able to find an essential set by solving $5\sim35$ linear programs and $2$ mixed integer linear programs. 
It takes from 4934 seconds ($N=100$) to 6847 seconds ($N=1000$) to solve one MILP (s-SCUC). When searching for support scenarios for the second-stage problem (a linear program), it takes $281 \sim 388$ seconds to solve one LP.
For those 8 out of 200 simulations, it takes an extra $20$ hours to find an irreducible set using Algorithm \ref{alg:find-irreducible-set}. This computation time can be greatly reduced by tricks such as choosing appropriate starting points\footnote{For example, when removing scenarios $s$ and $t$ consecutively in Algorithm \ref{alg:find-irreducible-set}, the solution $\opt_\mathbb{A}(\N-s)$ is feasible to $\SP(\N-s-t)$ thus can serve as a good starting point.}.
% subsection finding_essential_sets_for_s_scuc (end)

% section case_study (end)

\section{Discussions} % (fold)
\label{sec:discussions}

\subsection{Cardinality of Essential Sets} % (fold)
\label{sub:cardinality_of_essential_sets}
Fig. \ref{fig:sc-SCUC-case118iit-support-scenario-number} compares the cardinalities of essential sets for three cases: (a) c-SCUC with $N-1$ contingencies but without transmission constraints, results of case (a) are obtained from \cite{geng_chance-constrained_2019}); (b) c-SCUC with transmission constraints but without $N-1$ contingencies; and (c) c-SCUC with both transmission constraints and $N-1$ contingencies.
\begin{figure}[htbp]
  \centering
  \includegraphics[width=\linewidth]{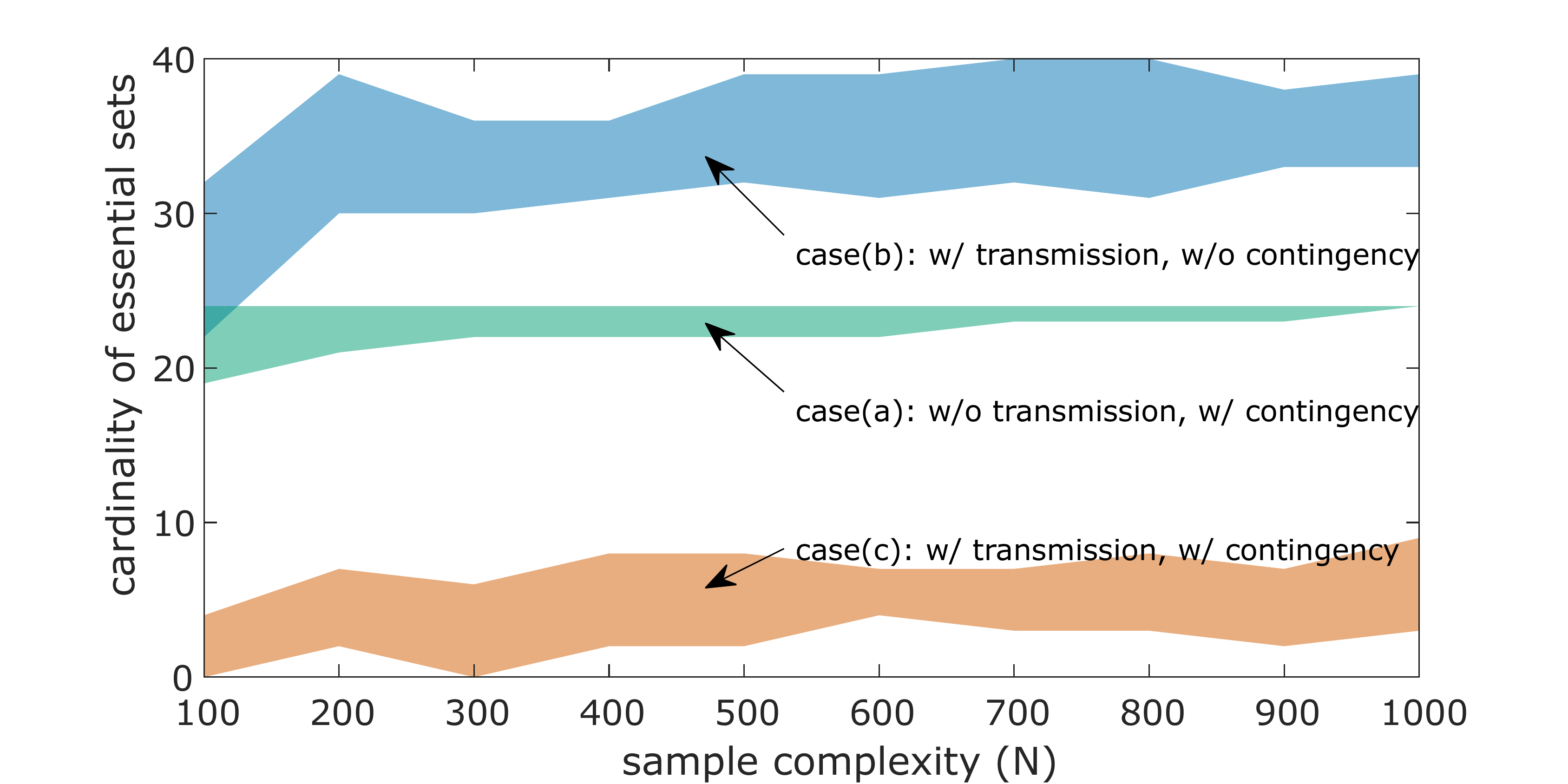}
  \caption{Cardinality of Essential Sets.}
  \label{fig:sc-SCUC-case118iit-support-scenario-number}
\end{figure}
Case (a) is the simplest, in \cite{geng_chance-constrained_2019} we show that the scenario problem for unit commitment satisfies the non-degeneracy assumption \ref{ass:non-degeneracy}, and the cardinality of essential sets is bounded by the number of snapshots $n_t$, i.e., $|\Sc| \le n_t = 24$ in Fig. \ref{fig:sc-SCUC-case118iit-support-scenario-number}. Cases (b) and (c) include transmission capacity constraints. As demonstrated in Section \ref{sub:s_scuc_is_degenerate}, s-SCUC could be degenerate with transmission constraints. Theoretically speaking, the cardinality of essential sets might be unbounded for non-convex problems. \FR{Fig. \ref{fig:sc-SCUC-case118iit-support-scenario-number} shows that the cardinality of essential sets is greatly smaller than the number of decision variables. For example, s-SCUC has about $4000$ binary variables and around $75000$ continuous decision variables, while the cardinalities of essential sets are $30\sim40$ in case (b) and $0\sim10$ in case (c).} This observation implies that the number of scenarios $N$ required could be much smaller than expected. 

Another interesting observation is that including $N-1$ contingency constraints reduces $|\Es|$. This observation has two implications. First, $N-1$ contingency constraints not only protect the system from unexpected device failures, they also help reduce the impacts of uncertainties from renewables. Second, including $N-1$ contingency constraints could help reduce sample complexity. Similar with the observations in \cite{modarresi_scenario-based_2018}, this observation indicates that the scenario approach might be of practical use.

% \vspace{-0.5cm}
\subsection{From Posterior to Prior Guarantees} % (fold)
\label{sub:from_posterior_to_prior_guarantees}
Theorem \ref{thm:scenario_theory_nonconvex_posterior} gives \emph{posterior} guarantees on the feasibility of solutions, namely, we calculate $\epsilon(N,k,\beta)$ \emph{after} obtaining the solution $\opx(\N)$.
Lemma \ref{lem:monotonicity_specified_epsilon_function} proves that the $\epsilon(N,k,\beta)$ function in \eqref{eqn:specified_epsilon_function} is monotone in $N$ and $k$. This implies that we can obtain prior guarantees. In other words, if the cardinality of essential sets is proved to be at most $h$ ($|\Es| \le h$), then we can compute the smallest $\hat{N}$ (e.g., using binary search) such that $\epsilon \ge 1 - \Big( \frac{\beta}{ \hat{N} \binom{\hat{N}}{h}} \Big)^{\frac{1}{\hat{N}-h}}$
holds for given $\epsilon$ and $\beta$. Then the solution $\opx_{\mathbb{A}}(\N)$ to the scenario problem using $\hat{N}$ scenarios has the guarantee $\mathbb{P}(\mathbb{V}(\opx_{\mathbb{A}}(\N)) \le \epsilon) \ge 1 - \beta$. This prior guarantee holds \emph{before} solving the scenario problem with $\hat{N}$ scenarios. If a rigorous bound $h$ on $|\Es|$ can be proved, then there is no need to numerically search for essential sets. This is particularly attractive compared with posterior guarantees. \FR{One example is the chance-constrained unit commitment (without transmission line limits) in \cite{geng_chance-constrained_2019}, in which we show that the cardinality of essential sets is bounded by the number of snapshots $n_t$, i.e., $|\Sc| \le n_t = 24$, also shown in case (a) of Fig. \ref{fig:sc-SCUC-case118iit-support-scenario-number}.
Although general non-trivial bounds on $|\Es|$ may not exist, results of \cite{geng_chance-constrained_2019} indicate that bounding $|\Es|$ is possible for structured non-convex problems.  }

% subsection from_posterior_to_prior_guarantees (end)
% \vspace{-0.5cm}
\subsection{Optimality Gap vs Violation Probability} % (fold)
\label{sub:optimality_gap_versus_violation_probability}
\FR{Unlike Theorem \ref{thm:exact_feasibility_scenario_approach}, which only holds for the global optimal solution, the non-convex scenario approach theory (Theorem \ref{thm:scenario_theory_nonconvex_posterior}) works for any feasible solution to the scenario problem. In Figs. \ref{fig:sc-SCUC-case118iit-fig_eps01_fixx} and \ref{fig:sc-SCUC-case118iit-fig_epsgap_fixx}, we compare the performance of a suboptimal solution\footnote{For every $N$ and every Monte-Carlo run, both solutions are using the same scenarios.)} with the optimal solution (MIP Gap less than $0.01\%$) in Section \ref{ssub:cost_security_trade_off}. The suboptimal solution is obtained by fixing all $z_{i}^t = 1,u_{i}^t=0, v_{i}^t=0, \forall i, \forall t$ and solve the s-SCUC problem. In other words, the suboptimal solution does not take the optimal commitment of generators into account, it only optimizes the dispatch $(g,r)$ variables.
}

\begin{figure}[htbp]
  \centering
  \includegraphics[width=\linewidth]{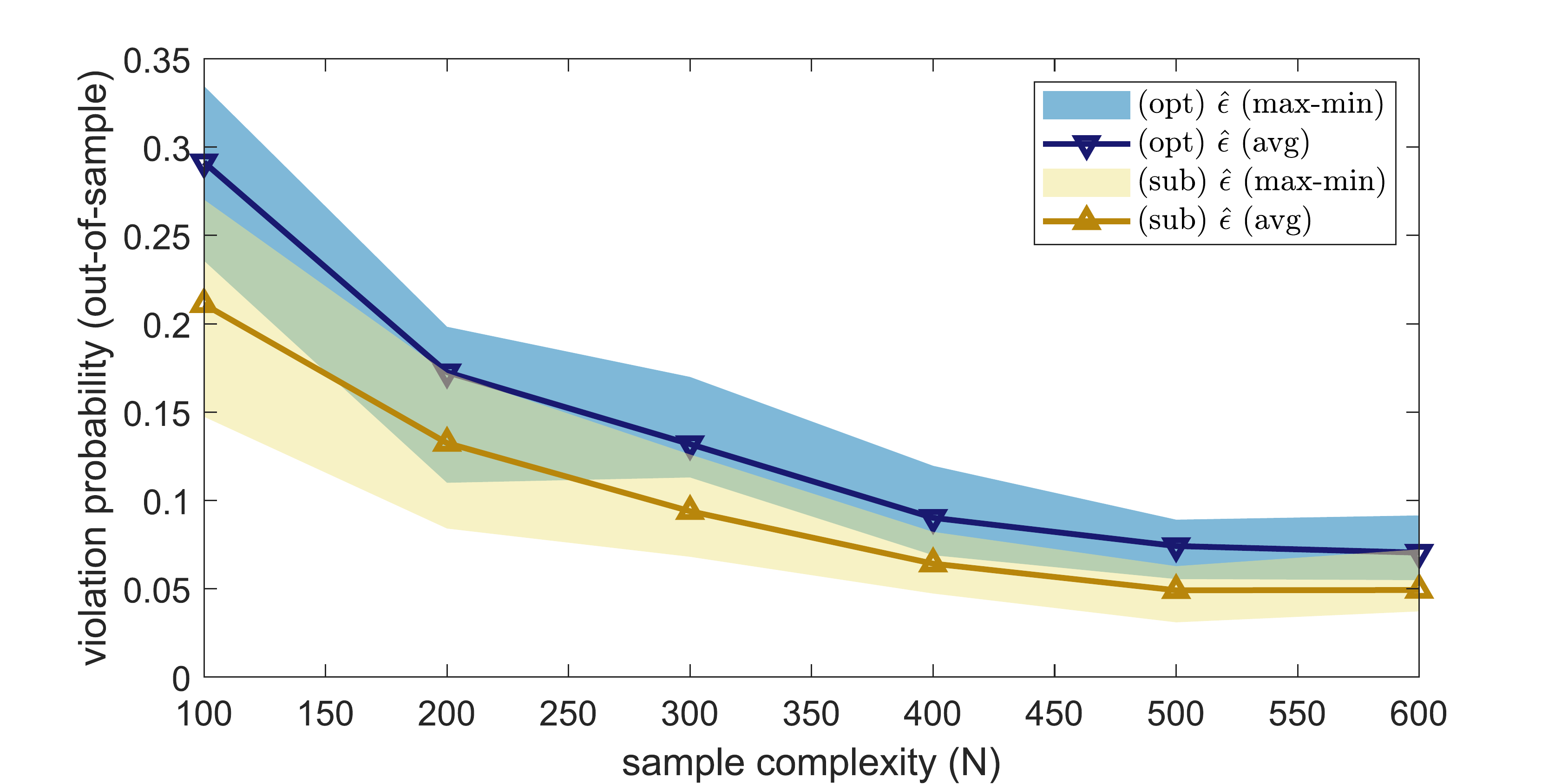}
  \caption{\FR{Out-of-sample Violation Probabilities of the Optimal Solution (opt) and a Suboptimal Solution (sub)}}
  \label{fig:sc-SCUC-case118iit-fig_eps01_fixx}
\end{figure}
\FR{Fig. \ref{fig:sc-SCUC-case118iit-fig_eps01_fixx} compares the out-of-sample violation probabilities. Using the same number of scenarios $N$, suboptimal solutions always have smaller violation probabilities (more secure or conservative). If violation probabilities of solutions are guaranteed to be acceptable ranges, then we should always pursue the optimal solution with smaller objective values to reduce conservativeness.}
\begin{figure}[htbp]
  \centering
  \includegraphics[width=\linewidth]{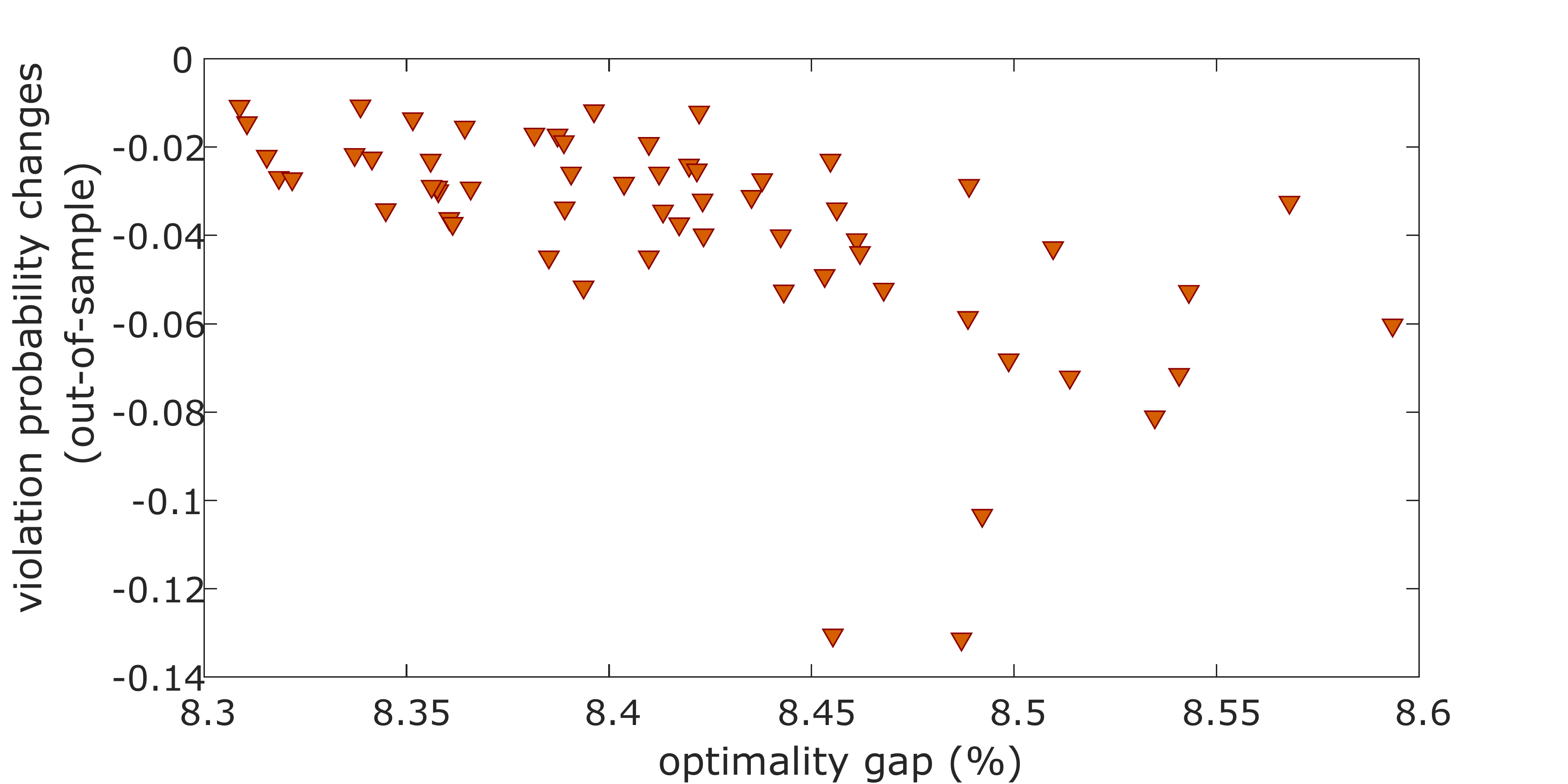}
  \caption{\FR{Optimality Gap vs Violation Probability. 60 dots in this figure represent 60 instances of s-SCUC solved.}}
  \label{fig:sc-SCUC-case118iit-fig_epsgap_fixx}
\end{figure}

\FR{Fig. \ref{fig:sc-SCUC-case118iit-fig_epsgap_fixx} further examines the relationship between optimality and security. The x-axis is the optimality gap, i.e., the percentage differences between suboptimal solutions and optimal solutions. The y-axis is the difference between violation probabilities of suboptimal and optimal solutions. Negative values indicate that violation probabilities of suboptimal solutions are always smaller than those of optimal solutions. There is also a subtle trend that larger optimality gaps lead to bigger decrease in violation probabilities.}

% \subsection{Solving the Scenario Problem} % (fold)
% \label{sub:solving_the_scenario_problem}
% \XG{say sth about how the theory could help to solve the problem.}
% subsection solving_the_scenario_problem (end)
% subsection optimality_gap_versus_violation_probability (end)
% section discussions (end)
% \vspace{-0.3cm}
\section{Concluding Remarks} % (fold)
\label{sec:concluding_remarks}
% This paper solves chance-constrained SCUC via the scenario approach and obtains rigorous theoretical guarantees on the solution. We demonstrate the structural properties of (possibly non-convex) general scenario problems. To obtain the tightest theoretical guarantees for chance-constrained SCUC, we design efficient algorithms to search for essential sets by exploiting the salient structures of SCUC.
\FR{This paper studies a core problem in the scenario approach theory, i.e., efficiently identifying essential sets for general scenario problems. For convex problems, we prove that the sparsest dual solution of the scenario problem could pinpoint the essential set. For non-convex problems, we provide conditions where simple algorithms based on definitions could return the essential set when the scenario problem is non-degenerate.}
We use chance-constrained Security-constrained Unit Commitment as a numerical example.
Case studies on an IEEE benchmark system show that the essential scenario set is only a small subset of all scenarios. This implies that we can obtain relatively robust solutions (i.e., small $\epsilon$) using only a moderate number of scenarios. Furthermore, we observe that some power engineering practices (e.g., $N-1$ criteria) can help us reduce the number of scenarios needed while maintaining the same level of risk.

Future work includes (1) reducing conservativeness by improving the complexity bound in Theorem \ref{thm:scenario_theory_nonconvex_posterior}; (2) a more systematic approach to compute the sparsest dual solution of convex scenario problems; and (3) investigating the performance of the (non-convex) scenario approach on larger-scale realistic power systems.
% section concluding_remarks (end)

\appendices

\input{proofs}

\input{algorithms}

\input{examples}

\bibliographystyle{IEEEtran}
\bibliography{myreferences}
\end{document}

%% file: proofs.tex
\vspace{-0.3cm}
\section{Proofs} % (fold)
\label{sec:proofs}
The proof of Lemma \ref{lem:monotonicity_specified_epsilon_function} is through basic inequalities, we omitted it to meet the length requirement (max 13 pages). Complete details about the proof of Lemma \ref{lem:monotonicity_specified_epsilon_function} is available in the supplementary material and at \cite{geng_computing_2019}.

\begin{proof}[Proof of Lemma \ref{lem:support_is_subset_invariant_nonconvex} \cite{calafiore_random_2010}]
For the purpose of contradiction, we assume that there is a scenario $s \in \Sc$ but $s \notin \Iv$. According to the definition of support scenarios, $\opt_\mathbb{A}(\N-s) < \opt_\mathbb{A}(\N)$. However, Assumption \ref{ass:monotone_algorithm} claims that removing scenarios will not increase the optimal objective value and $\Iv \subseteq \N-s$, we have $\opt_\mathbb{A}(\N-s) \ge \opt_\mathbb{A}(\Iv) =\opt_\mathbb{A}(\N)$, which causes a contradiction.
\end{proof}

\begin{proof}[Proof of Lemma \ref{lem:removal_non_support}]
$k \notin \Sc$ and $s \in \Sc$ give $\opt(\N - k) = \opt(\N)$ and $\opt(\N - s) < \opt(\N)$, respectively. Assumption \ref{ass:monotone_algorithm} shows $\opt(\N - k - s) \le \opt(\N - s)$. Hence, it holds that
\begin{small}
\begin{multline*}
  \opt(\N - k - s) \le \opt(\N - s) < \opt(\N) = \opt(\N - k),\forall s \in \Sc(\N)
\end{multline*}
\end{small}
then $s$ is a support scenario for $\SP(\N - k)$, thus $\Sc(\N) \subseteq \Sc(\N - k)$.
\end{proof}

\begin{proof}[Proof of Lemma \ref{lem:inactive_are_not_support_scenario}]
\FR{
For the purpose of contradiction, we assume $i \in \mathcal{S} \cap \mathcal{P}$. Since $i$ is a support scenario, there is a new optimal solution $x_{\mathcal{N}-i}^*$ with $c^\intercal x_{\mathcal{N}-i}^* < c^\intercal x_{\mathcal{N}}^*$. Let $0 < \alpha < 1$ be a positive number, and we consider the solution $\hat{x} = \alpha x_{\mathcal{N}-i}^* + (1- \alpha) x_{\mathcal{N}}^*$. It is obvious that 
\begin{equation*}
c^\intercal \hat{x} < \alpha c^\intercal x_{\mathcal{N}}^* + (1- \alpha) c^\intercal x_{\mathcal{N}}^* = c^\intercal x_{\mathcal{N}}^*
\end{equation*} 
Next we show that $\hat{x}$ is feasible to $\SP(\N)$ if $\alpha$ is small enough. Since $i \in \mathcal{P}$, there exists a positive number $\delta$ such that $f(x_{\mathcal{N}}^*, \xi^i) \le - \delta$. Using the fact that $f(x,\xi)$ is convex in $x$ and convex functions are also Lipschitz (with Lipschitz constant $L$), we can derive an upper bound $\overline{\alpha}(L, \delta)$, such that for any $0 < \alpha \le \overline{\alpha}(L,\delta)$, we have $f(x_{\N}^*, \xi^i) \le - \delta \le f(\hat{x}, \xi^i) \le 0$.
For other constraints $f(x,\xi^j) \le 0$ ($j\ne i$), it is easy to see that $\hat{x}$ is also feasible because it is the convex combination of two feasible points. Therefore $\hat{x}$ is another feasible solution to $\SP(\N)$ but with smaller objective $c^\intercal \hat{x} < c^\intercal x_{\mathcal{N}}^*$. This causes contradiction.
}
\end{proof}

\begin{proof}[Proof of Theorem \ref{thm:support_scenario_multiplier}]
% \paragraph{Prove (1)} % (fold)
% \label{par:prove_1}
\FR{\emph{(Preparation)}
The Lagrange dual function $D_{\N}(\mu,\lambda)$ of $\SP(\N)$ is:
\begin{equation}
D_{\N}(\mu,\lambda) = \inf_{x} \Big( c^\intercal x + \sum_{\iota=1}^{N} (\mu^\iota)^\intercal f(x,\xi^\iota) + \lambda^\intercal g(x) \Big )
\end{equation}
The Lagrange dual problem is $\max_{\mu \ge 0,\lambda \ge 0} D_{\N}(\mu,\lambda)$, and we use $\lambda_{\N}^*$ and $\mu_{\N}^* = \{\mu_{\N}^{i,*}\}_{i=1}^{N}$ to denote its optimal solution.
By Assumption \ref{ass:feasibility_uniqueness} (feasibility), we know that $\SP(\N)$ has a strictly feasible solution, thus Slater's condition holds and $D(\mu_{\mathcal{N}}^*,\lambda_{\mathcal{N}}^*) = c^\intercal x_{\mathcal{N}}^*$ by strong duality. 
We then consider the Lagrange dual problem of $\SP(\N-j)$. Instead of directly removing constraints $f(x, \xi^j) \le 0$, we consider the following relaxed version
\begin{equation}
f(x, \xi^j) \le M \mathbf{1}.
\end{equation}
When constant $M$ is sufficiently large, it is as if constraints $f(x, \xi^j) \le 0$ are removed. The associated dual problem is 
\begin{equation}
D_{\N-j}(\mu,\lambda) = D_{\N}(\mu,\lambda) - M \mathbf{1}^\intercal \mu^{j}
\end{equation}
The optimal dual solution of $\SP(\N-j)$ is denoted by $\lambda_{\N-j}^*$ and $\mu_{\N-j}^* = \{\mu_{\N-j}^{i,*}\}_{i=1}^{N}$. Note that $\| \mu_{\N-j}^{j,*} \| = 0$, otherwise $D_{\N-j}(\mu,\lambda)$ will be dominated by $M \mathbf{1}^\intercal \mu^{j}$, which is unbounded when $M$ is arbitrarily large, which leads to contradiction.
}

\FR{
\emph{We first prove (1).} For the purpose of contradiction, we assume $j$ is a support scenario but with $\| \mu_{\N}^{j,*} \| = 0$. Since $j$ is a support scenario, from strong duality, we know that
\begin{equation*}
D_{\N}(\mu_{\N}^*, \lambda_{\N}^*) = c^\intercal x_{\N}^* > c^\intercal x_{\N-j}^* = D_{\N-j}(\mu_{\N-j}^*, \lambda_{\N-j}^*)
\end{equation*}
If $\| \mu_{\N}^{j,*} \| = 0$, then $(\lambda_{\N}^*,\mu_{\N}^*)$ is a feasible solution to $\max_{\mu \ge 0,\lambda \ge 0} D_{\N-j}(\mu,\lambda)$, therefore
\begin{multline*}
D_{\N-j}(\mu_{\N-j}^*, \lambda_{\N-j}^*) \ge D_{\N-j}(\mu_{\N}^*, \lambda_{\N}^*) \\
= D_{\N}(\mu_{\N}^*, \lambda_{\N}^*) - M \mathbf{1}^\intercal \mu_{\N}^{j,*} = D_{\N}(\mu_{\N}^*, \lambda_{\N}^*) \\ > D_{\N-j}(\mu_{\N-j}^*, \lambda_{\N-j}^*)
\end{multline*}
which is a contradiction.
}

\FR{
\emph{Next we prove (2)} by constructing an optimal dual solution with $\|\mu^{j,*}\| = 0$. If $\xi^j$ is \emph{not} a support scenario, then
\begin{equation}
D_{\N}(\mu_{\N}^*, \lambda_{\N}^*) = c^\intercal x_{\N}^* = c^\intercal x_{\N-j}^* = D_{\N-j}(\mu_{\N-j}^*, \lambda_{\N-j}^*)
\end{equation}
by Slater's condition and strong duality. We could assign $\hat{\mu}_{\N}^{\iota,*} = \mu_{\N-j}^{\iota,*}$ for $\iota \ne j$ and let $\hat{\mu}_{\N}^{j,*} = \mathbf{0}$. For the constructed solution $\hat{\mu}_{\N}^*$, we have
\begin{multline*}
D_{\N}(\hat{\mu}_{\N}^*, \lambda_{\N-j}^*) = D_{\N-j}(\mu_{\N-j}^*, \lambda_{\N-j}^*) = c^\intercal x_{\N}^*
\end{multline*}
Clearly this is one optimal dual solution to $\SP(\N)$.
}

\FR{
\emph{Next we prove (3).} First, we solve the dual problem of $\SP(\Sc)$ (with support scenarios only) and obtain optimal solution $(\lambda_{\Sc}^*,\mu_{\Sc}^*)$. Note that $\|\mu_{\mathcal{S}}^{\iota,*}\| > 0$ for any $\iota \in \mathcal{S}$ according to (1). Next, we assign $\mu_{\N}^{\iota,*} = \mu_{\Sc}^{\iota,*}$ for $\iota \in \Sc$, and assign $\mu_{\N}^{\iota,*} = \mathbf{0}$ for $\iota \notin \Sc$. Similar to the proof of (2), this constructed solution $(\lambda_{\Sc}^*,\mu_{\N}^*)$ is an optimal dual solution to $\SP(\N)$. With this constructed $\mu_{\N}^*$, we know that $\mathcal{M}(\mu_{\N}^*) \subseteq \Sc$. Combining with the claim $\Sc \subseteq \mathcal{M}(\mu_{\N}^*)$ in (1), we have  $\mathcal{M}(\mu_{\N}^*) = \Sc$.
}

\FR{(4) is obvious from (3).}
\end{proof}

\begin{proof}[Proof of Corollary \ref{cor:nondegenerate_uniqueness}]
We first prove (1), that is $\SP(\N)$ has a unique essential set if it is non-degenerate (similar with the proof of Lemma 2.11 in \cite{calafiore_random_2010})). From Lemma \ref{lem:support_is_subset_invariant_nonconvex}, an essential set can be written as $\Es = \Sc \cup \mathcal{Y}$ where $\mathcal{Y} \subseteq (\N - \Sc)$. The support set $\Sc$ is invariant because of the non-degeneracy of $\SP(\N)$ by assumption. Since $\Es$ is the invariant set of minimal cardinality, we can let $\mathcal{Y} = \emptyset$ and $\Sc$ is the essential set. The support set $\Sc$ is unique by definition, this implies the uniqueness of the essential set $\Es$ for non-degenerate $\SP(\N)$.

We then prove (2). Lemma \ref{lem:support_is_subset_invariant_nonconvex} shows that $\Sc \subseteq \Ir$, we only need to show $\Ir \subseteq \Sc$ when $\SP_\mathbb{A}(\N)$ is non-degenerate.
For the purpose of contradiction, we assume there exists $s \in \Ir$ but $s \notin \Sc$. By hypothesis ($s \notin \Sc$), we have $\Sc \subseteq \Ir - s$ (Lemma \ref{lem:support_is_subset_invariant_nonconvex}). The monotonicity assumption \ref{ass:monotone_algorithm} gives $\opt_\mathbb{A}(\Sc) \le \opt_\mathbb{A}(\Ir -s)$. Since $\Ir$ is irreducible, we have $\opt_\mathbb{A}(\Ir - s) < \opt_\mathbb{A}(\Ir)$. $\SP_\mathbb{A}(\N)$ is non-degenerate and $\Ir$ is invariant gives $\opt_\mathbb{A}(\Ir) = \opt_\mathbb{A}(\N) = \opt_\mathbb{A}(\Sc)$. Combining the results above, we have
\begin{equation*}
  \opt_\mathbb{A}(\Sc) \le \opt_\mathbb{A}(\Ir -s) < \opt_\mathbb{A}(\Ir) = \opt_\mathbb{A}(\N) = \opt_\mathbb{A}(\Sc),
\end{equation*}
which is clearly a contradiction. Therefore $\Sc = \Ir$.
\end{proof}

\begin{proof}[Proof of Theorem \ref{thm:nondegenerate_unique_essential_irreducible}]
(1) $\Rightarrow$ (2) is proved in Corollary \ref{cor:nondegenerate_uniqueness}. And (2) $\Rightarrow$ (3) is obvious, since the essential set $\Es$ is irreducible. If there is only one irreducible set, then it is the essential set.

Lastly, we prove (3) $\Rightarrow$ (1). We prove $\SP(\N)$ being degenerate implies the essential set is not unique (equivalent with the statement that $\SP(\N)$ is non-degenerate if it has a unique essential set). Suppose $\SP(\N)$ is degenerate, i.e. $\opt(\Sc) < \opt(\N)$. Consider an essential set $\Es = \Sc \cup \mathcal{T}$ (Lemma \ref{lem:support_is_subset_invariant_nonconvex}), where $\mathcal{T}$ is non-empty and $k \in \mathcal{T}$. Consider the scenario problem $\SP(\N - k)$, and $\opt(\N-k) = \opt(\N)$ because $k\notin \Sc$. We also know that $\Sc$ is contained in any essential set of $\SP(\N - k)$ by Lemma \ref{lem:removal_non_support}, i.e. $\Es(\N-k) = \Sc \cup \hat{\mathcal{T}}$. And $\hat{T}$ has to be non-empty \footnote{Otherwise $\opt(\Sc) = \opt(\Es(\N-k)) = \opt(\N-k)=\opt(\N)$, which contradicts with the hypothesis that $\SP(\N)$ is degenerate.}. Then $\opt(\Sc \cup \hat{\mathcal{T}}) = \opt(\N-k) = \opt(\N)$, therefore $\Sc \cup \hat{\mathcal{T}}$ must contain at least one essential set that is different from $\Sc \cup \mathcal{T}$ (because $k \in \mathcal{T}$ and $k \notin \hat{\mathcal{T}}$). Therefore $\SP(\N)$ has more than one essential set when it is degenerate.
\end{proof}
% \begin{lem}
% Moreover, let $\mathcal{M} \subseteq \N$, if $\SP(\N)$ is non-degenerate, then $\opt_\mathbb{A}(\N) = \opt_\mathbb{A}(\mathcal{M})$ if and only if $\Es(\N) = \Es(\mathcal{M})$
% \end{lem}
% This lemma is not useful for our work.
% section proofs (end)

% \begin{proof}[Proof of Corollary \ref{cor:convex_SP_unique_dual}]
% From Lemma xx, we know $\mathcal{S} \subseteq \mathcal{Q}$. So we only need to prove $\mathcal{Q} \subseteq \mathcal{S}$. If the dual solution is unique, then $\mathcal{M} = \mathcal{S}$ according to Theorem \ref{thm:support_scenario_multiplier}. If $i \in \mathcal{M}$, i.e., $\|\mu^{i,*}\| > 0$, then there exists a $j$ such that $\mu_j^{i,*} > 0$. By complementary slackness, we know that the $j$th constraint of scenario $i$ is binding, $f_j(x, \xi^i) = 0$, so $j \in \mathcal{Q}$

% Therefore $\mathcal{S} \in \mathcal{Q}$.
% \end{proof}

\begin{proof}[Proof of Lemma \ref{lem:two_stage_support}]
We first prove (1). The case that $\hat{\Sc} = \emptyset$ is trivial. For the case that $\hat{\Sc}$ contains at least one scenario $s \in \hat{\Sc}$. Solving the 2nd stage problem with $s$ removed gives a different optimal solution $\hat{x}$ with $c_x^\intercal \hat{x} < c_x^\intercal x^*$. Clearly $(\hat{x},y^*)$ is a feasible solution to $\SP(\N - s)$, with 
\begin{equation}
  c_y^\intercal y^* + c_x^\intercal \hat{x} < c_y^\intercal y^* + c_x^\intercal x^*
\end{equation}
therefore $s$ is a support scenario for $\SP(\N)$ and $\hat{\Sc} \subseteq \Sc$.

We then prove (2). By Assumption \ref{ass:monotone_algorithm}, we know that $\opt_\mathbb{A}(\hat{\Sc}) \le \opt_\mathbb{A}(\Sc) \le \opt_\mathbb{A}(\N)$ since $\hat{\Sc} \subseteq \Sc \subseteq \N$. If $\hat{\Sc}$ is invariant, i.e. $\opt_\mathbb{A}(\hat{\Sc}) = \opt_\mathbb{A}(\N)$, then $\opt_\mathbb{A}(\N) \le \opt_\mathbb{A}(\hat{\Sc}) \le \opt_\mathbb{A}(\Sc) \le \opt_\mathbb{A}(\N)$ gives $\opt_\mathbb{A}(\Sc) = \opt_\mathbb{A}(\N)$, therefore $\SP(\N)$ is non-degenerate.
\end{proof}

%% file: algorithms.tex
\vspace{-0.5cm}
\section{Algorithms} % (fold)
\label{sec:algorithms}
% \subsection{Algorithms} % (fold)
% \label{sec:algorithms}

\vspace{-0.3cm}
\begin{algorithm}[H]
\begin{algorithmic}[1]
\STATE Compute $\opx_{\mathbb{A}}(\N)$ by solving $\SP_{\mathbb{A}}(\N)$. Set $\mathcal{I} \leftarrow \mathcal{N}$.
\FOR{$i \in \mathcal{N}$}
\STATE Compute $\opx_{\mathbb{A}}(\mathcal{I}-i)$ by solving $\SP(\mathcal{I}-i)$.
\IF{$\opt_{\mathbb{A}}(\mathcal{I}-i) = \opt_{\mathbb{A}}(\N)$} 
  \STATE $ \mathcal{I} \leftarrow \mathcal{I} - i$.
\ENDIF
\ENDFOR
\end{algorithmic}
\caption{Find an Irreducible Set $\mathcal{I}$ of $\SP_{\mathbb{A}}(\N)$}
\label{alg:find-irreducible-set}
\end{algorithm}

\vspace{-0.3cm}

\begin{algorithm}[H]
\begin{algorithmic}[1]
\STATE Compute $x_{\mathcal{N}}^*$ by solving $\SP(\N)$.
\STATE Set $\mathcal{S} \leftarrow \emptyset$.
\FOR{$i \in \mathcal{N}$}
\STATE Solve the scenario problem $\text{SP}_{\mathcal{N}-i}$ and compute $x_{\mathcal{N}-i}^*$.
\IF{$c^\intercal x_{\mathcal{N}-i}^* < c^\intercal x_{\mathcal{N}}^*$} 
  \STATE $ \mathcal{S} \leftarrow \mathcal{S} + i$.
\ENDIF
\ENDFOR
\end{algorithmic}
\caption{Find the Support Set $\mathcal{S}$ of $\SP(\N)$}
\label{alg:find-support-set}
\end{algorithm}

% \begin{algorithm}[H]
% \begin{algorithmic}[1]
% \STATE Initialize $k=0$, $\mathcal{O} = \emptyset$, $\mathcal{M} = \emptyset$ and $\rho^* = +\infty$
% \WHILE {$\rho^* \ge 0$}
% \STATE Let $x^{(k)}$ be the optimal solution to 
% \begin{subequations}
% \begin{align}
% x^{(k)} = \arg \min_{x}~& c^\intercal x \\
% \text{s.t.}~& f_j(x, \xi^{i}) \le 0,~(i,j) \in \mathcal{O} \\
% & g(x) \le 0
% \end{align}
% \end{subequations}
% \STATE Find the scenario $i^*$ and constraint $j^*$ with the largest constraint violation $\rho^* = f_{j^*}(x, \xi^{i^*})$
% \begin{equation}
% \rho^* = \max_{j = \{1,2,\cdots,m.\}} \max_{i \in \mathcal{N}/\mathcal{M} } \big\{  f(x, \xi^{i}) \big\}
% \end{equation}
% \STATE $\mathcal{O} \leftarrow \mathcal{O} \cup \{(i^*, j^*)\}$, $\mathcal{M} \leftarrow \mathcal{M} \cup \{i^*\}$ and $k \leftarrow k+1$.
% \ENDWHILE 
% \STATE Output $x^{(k+1)}$, $\mathcal{M}$ and $\mathcal{O}$
% \end{algorithmic}
% \caption{ \FR{Solve Scenario Problem} }
% \label{alg:solve-scenario-problem}
% \end{algorithm}

% Obviously Algorithm \ref{alg:solve-scenario-problem} converges in finite steps (worst case adding all $mN$ constraints). 

\vspace{-0.5cm}

\begin{algorithm}[H]
\begin{algorithmic}[1]
\STATE Solve $\SP_{\mathbb{A}}(\N)$ and compute the solution $(x^*, y^*)$.
% \STATE Set $\mathcal{L} \leftarrow \{\xi^1,\xi^2,\cdots,\xi^N\}$ and 
\STATE Fix $y=y^*$, find support scenarios $\mathcal{S}$ of the second-stage problem \eqref{opt:second_stage_scenario_problem}, e.g. using Algorithm \ref{alg:find-support-scenario-dual}.
% \begin{equation*}
%  \min_{x \in \mathcal{X},(x,y^*) \in \mathcal{H}}~ c_x^\intercal x~\text{s.t.}~x \in \cap_{i=1}^N \mathcal{U}_i
% \end{equation*}
% \STATE Check if $\SP_{\mathbb{A}}(\N)$ is non-degenerate by checking $\mathbb{A}(\mathcal{S}) = \mathbb{A}(\mathcal{N})$
\IF{$\opt_{\mathbb{A}}(\mathcal{S}) = \opt_{\mathbb{A}}(\mathcal{N})$}
  \STATE $\SP_{\mathbb{A}}(\N)$ is non-degenerate and $\mathcal{S}$ is the essential set.
\ELSE
  \STATE 
  	$\SP_{\mathbb{A}}(\N)$ is degenerate, the best we can find is an irreducible set, e.g. using Algorithm \ref{alg:find-irreducible-set}.
  % \STATE finding the remaining scenarios (e.g. Algorithm \ref{alg:find-irreducible-set})
\ENDIF
\end{algorithmic}
\caption{For the two-stage scenario problem \eqref{opt:two-stage-scenario-problem} }
\label{alg:find-support-subsample-structural}
\end{algorithm}
% subsection algorithms (end)

%% file: examples.tex
\section{Illustrative Examples} % (fold)
\label{sec:settings_of_the_3_bus_system}
% \vspace{-0.3cm}
% \section{Illustrative Example in Figure \ref{fig:primal-example}} % (fold)
% \label{sub:illustrative_example_in_figure_fig:primal-example}

% section illustrative_example_in_figure_fig:primal-example (end)
\subsection{Illustrative Example in Figure \ref{fig:primal-example}} % (fold)
\label{sub:illustrative_example_in_figure_fig:primal-example}
\FR{We consider the scenario problem in \eqref{opt:scenario-problem-ex-primal}.
\begin{small}
\begin{subequations}
\label{opt:scenario-problem-ex-primal}
\begin{align}
\min_{x_1, x_2}~& x_2 \\
\text{s.t.}~& x_2 \ge x_1,\hspace{0.95cm} x_2 \ge -x_1 &~\text{(scenario 1)} \\
& x_2 \ge 2x_1 - 1, ~x_2 \ge 3 - 2x_1 &~\text{(scenario 2)} \\
& x_1 \ge 0,~x_2 \ge 0
\end{align}
\end{subequations}
\end{small}
The constraint $f(x,\xi) \le 0$ takes the following form with two scenarios: $\xi^1= 1$ and $\xi^2= 2$.
\begin{small}
\begin{align*}
f_1(x,\xi) =  \xi x_1 - x_2 - (\xi-1) \le 0 & \\
f_2(x,\xi) = -\xi x_1 - x_2 + 3(\xi-1) \le 0 &
\end{align*}
\end{small}
The dual of \eqref{opt:scenario-problem-ex-primal} is
\begin{small}
\begin{subequations}
\begin{align}
\max_{\mu_{11},\mu_{12},\mu_{21},\mu_{22}}~& -\mu_{21} + 3 \mu_{22} \\
\text{s.t.}~& \mu_{11} - \mu_{12} - 2 \mu_{21} + 2 \mu_{22} \le 0 \\
& \mu_{11} + \mu_{12} + \mu_{21} + \mu_{22} \le 1 \\
& \mu_{11} \ge 0,~\mu_{12} \ge 0,~\mu_{21} \ge 0,~\mu_{22} \ge 0.
\end{align}
\end{subequations}
\end{small}
The optimal dual solution is $\mu_{11} = 0, \mu_{12} \ge 0$, $\mu_{21} = 0.5-0.75 \mu_{12}$ and $\mu_{22} = 0.5-0.25 \mu_{12}$. Note $\mu_{12}$ can take any positive value so the dual solution is not unique.
}
% subsection illustrative_example_in_figure_fig:primal-example (end)

\subsection{Settings of the 3-bus System in Section \ref{sub:s_scuc_is_degenerate}} % (fold)
\label{sub:settings_of_the_3_bus_system}
All settings of the 3-bus system can found in Table \ref{tab:case3_settings}.
\begin{table*}[tb]
  \caption{Settings of the 3-bus System}
  \label{tab:case3_settings}
  \centering

  \begin{tabular}{ccccc|ccccc}
  \hline

  \hline
  \multicolumn{5}{c|}{Line Data} & \multicolumn{5}{c}{Generator Data} \\
  \hline
  Line No. & From Bus & To bus  & Reactance (p.u.) & Capacity (MW) & Gen No. & Bus & Min & Max & Marginal Cost \\
  \hline
  1 & 1 & 2 & 1.0 & 20  & 1 & 1 & 20 & 100 & 1 \\
  2 & 1 & 3 & 1.0 & 100 & 2 & 2 & 20 & 90 & 100 \\
    \cline{6-10}
  3 & 2 & 3 & 1.0 & 100 \\
  \hline

  \hline
  \multicolumn{5}{c|}{Load Data (MW)} & \multicolumn{5}{c}{Wind Data (MW)} \\
  \hline
  Bus & Forecast  & Error 1 & Error 2 & Error 3 & Bus & Forecast & Error 1 & Error 2 & Error 3  \\
  \hline
  3 & 110 & 11 & -30 & -35 & 2 & 30 & 6 & -15 & -25 \\
  \hline

  \hline
  \end{tabular}
\end{table*}